\documentclass[11pt,tbtags,reqno]{article}
\usepackage[margin=1.3in]{geometry}
\usepackage{amssymb}
\usepackage{amsmath}
\usepackage{amsthm}
\usepackage[swedish, english]{babel}
\usepackage{psfrag}   
\usepackage{graphicx}
\usepackage[retainorgcmds]{IEEEtrantools}
\usepackage{bbm}
\usepackage[square, comma, numbers, sort&compress]{natbib}
\usepackage{enumerate}
\usepackage{hyperref}
\usepackage[colorinlistoftodos, textsize=tiny, backgroundcolor=yellow
, linecolor=green]{todonotes}
\usepackage{sidecap}
\usepackage{float}

\newcommand{\E}{\mathbb E}
\newcommand{\R}{\mathbb R}
\newcommand{\N}{\mathbb N}

\newcommand{\F}{\mathcal F}

\newcommand{\ep}{\epsilon}

\newcommand{\ud}{\,\mathrm{d}}
\newcommand{\vd}{\mathrm{d}}
\newcommand{\lt}{\left}
\newcommand{\rt}{\right}
\newcommand{\pt}{\partial}

\newcommand{\tend}{\rightarrow}
\newcommand{\into}{\rightarrow}

\newcommand{\gm}{\gamma}

\newcommand{\Dr}{X} 
\def\P{{\mathbb P}}
\def\Q{{\mathbb Q}}
\newcommand{\Ind}{\mathbbm{1}}

\newcommand{\supp}{\mathop{\mathrm{supp}}\nolimits}
\newcommand{\Var}{\mathop{\mathrm{Var}}\nolimits}

\numberwithin{equation}{section}
\theoremstyle{plain}
\newenvironment{example}[1][Example]{\begin{trivlist}
\item[\hskip \labelsep {\bf Example}]}{\end{trivlist}}
\newenvironment{remark}[1][Remark]{\begin{trivlist} \item[\hskip \labelsep {\bf Remark}]}
{\end{trivlist}}

\newtheorem{theorem}{Theorem}[section]
\newtheorem{lemma}[theorem]{Lemma}
\newtheorem{corollary}[theorem]{Corollary}
\newtheorem{proposition}[theorem]{Proposition}

\newtheorem{assump}[theorem]{Assumption}
\theoremstyle{definition}

\title{Optimal liquidation of an asset under drift uncertainty}
\author{Erik Ekstr\"om and Juozas Vaicenavicius}
\date{}

\begin{document}

\maketitle
\begin{abstract}
We study a problem of finding an optimal stopping strategy to liquidate an asset with unknown drift. Taking a Bayesian approach, we model the initial beliefs of an individual about the drift 
parameter by allowing
an arbitrary probability distribution to characterise the uncertainty about the drift parameter. 
Filtering theory is used to describe the evolution of the posterior beliefs about the drift once the price process is being observed. 
An optimal stopping time is determined as the first passage time of the posterior mean
below a monotone boundary, which can be characterised as the unique solution to a non-linear integral equation. We also study monotonicity properties with respect to the prior distribution and the asset volatility.
\end{abstract}

\section{Introduction}
It is an inevitable feature of human economic activity that prices of goods vary in time. Thus, naturally, a person participating in trade cares much about the best time to perform a transaction. Let us think about an individual who possesses an indivisible asset with price evolution $\{S_t\}_{t\geq0}$ and wants to sell it before time $T\geq 0$. 
Assuming a liquid market, how should the seller choose a selling time to 
maximise his/her profit from the sale?
Mathematically, the question is about finding a stopping time $\tau^*$, belonging to a set of admissible stopping times $\mathcal{T}_T$, such that
\begin{equation} \label{E:sell}
	\E[S_{\tau^*}] = \sup_{\tau \in \mathcal{T}_{T}} \E[S_\tau].
\end{equation}
A natural set of admissible stopping times $\mathcal{T}_{T}$ to consider is the set $\mathcal{T}^{S}_{T}$ of stopping times with respect to the price process $S$, i.e.~at any point in time, the decision whether to sell the asset or not must be based solely on the price history of $S$. 
Throughout this article we assume $\mathcal{T}_{T}=\mathcal{T}^{S}_{T}$.

In the context of the classical Black-Scholes model
\begin{equation} \label{E:Price}
	\ud S_t = \alpha S_t \ud t + \sigma S_t \ud W_t \,, 
\end{equation}
where $\alpha$, $\sigma$ are known constant parameters, the answer to the optimal selling question \eqref{E:sell} is straightforward:
if $\alpha>0$, then the optimal strategy is to sell at the terminal time $T$;
if $\alpha< 0$,  then the optimal strategy is to sell immediately, i.e. at time $0$;
if $\alpha=0$, then any stopping time $\tau$ is optimal.

However, in applications, the known constant drift assumption of the Black-Scholes model
is usually too strong. 
To obtain reasonable precision when estimating the drift one needs very long time-series, which are rarely available. An extreme example of the lack of data is a stock of an initial public offering (IPO) for which the price history simply does not exist. Furthermore, even in those few cases where enough past data is available, the benefit of accurate calibration of the historical drift is most likely to be outweighed by the model risk introduced. This is because most financial models, including Black-Scholes, are only plausible as short-term models; the simplistic constant parameter assumptions are non-viable over longer time periods. On the other hand, the assumption of known volatility parameter $\sigma$ is justifiable as it can be estimated, at least in theory, from an arbitrarily short observation period. 

Though the notoriety of the drift estimation pushed much of financial mathematics literature to focus on questions where the drift parameter can be avoided or at least does not play a crucial role (e.g.~risk-neutral pricing and hedging), in the optimal liquidation problem, the drift can have a noticeable effect. Figure \ref{Table}, containing the estimated Black-Scholes model parameters of a few famous IPOs over the first year since going public, suggests that it is unlikely that the price change in all these cases was due to the volatility alone, leading us to believe in the significance of the drift contribution that needs to be addressed.

\begin{figure}[h!] \label{Table}
\begin{center}
  \begin{tabular}{ l | c | c | c }
    IPO & $\log(S_{1}/S_{0})$ & $\hat \alpha$ & $\hat \sigma$ \\
    \hline \hline
    Amazon (1997)& 1.34 & 1.68 & 0.83 \\ \hline
    Google (2004) & 1.03 & 1.11 & 0.41 \\ \hline
    Facebook (2012)& -0.42 & -0.27 & 0.55 \\ \hline
	Vonage (2006) & -1.53 & -1.29 & 0.70 \\ \hline
  \end{tabular}
\end{center}
\caption{The estimates $\hat \alpha$ and $\hat \sigma$ of the drift $\alpha$ and the volatility $\sigma$ are calculated over the first year of an IPO using the daily closing prices. Data source: Google Finance. }
\end{figure}

Admission that due to unattainable calibration, in many situations, modelling a price by a geometric Brownian motion with a known constant drift is ill-suited is not a reason to give up modelling, but a mere indication that the model should be improved to incorporate extra factors. As the exact value of the drift parameter is unknown, we choose to model the inherent uncertainty about the drift by a probability distribution. More precisely, we extend the geometric Brownian motion model \eqref{E:Price} by replacing the constant drift $\alpha$ by a random variable whose distribution (called `a prior' in Bayesian statistics) incapsulates all the knowledge available to us concerning the uncertainty about the drift. As far as the volatility $\sigma$ is concerned, we stay with the known constant volatility assumption.

A potential practical application of this drift uncertainty modelling is in the optimal liquidation of an IPO share. A person possessing a share of an IPO has only beliefs about the drift of the price process as no past price data is available to calibrate the model. 
Though in this article we view the prior distribution as subjective beliefs whose origin we do not question, one can also think of transparent constructive approaches for choosing a prior. A possible approach in the IPO example is to use the empirical distribution of the returns of similar IPOs over the initial period of the same length as our investment horizon $T$. The similarity criteria could be the market sector, the country, the market share, etc.

For an agent interested in optimally liquidating only the idiosyncratic (i.e. stock-specific) component of a stock price, the sequential procedure is even more beneficial due to reduced volatility (relatively high volatility in Figure \ref{Table} could be seen as diminishing the advantage of the sequential liquidation procedure). A simple structural Black-Scholes model including an idiosyncratic and a market factor  (e.g.~see \cite{BH11}) suggests that the idiosyncratic price component 
\[
S_{t}/I_{t} \propto e^{(\mu' - \sigma'^{2}/2)t + \sigma' W'_{t}}.
\]
Here $I_{t}$ is a large-basket index representing the market factor, $\mu', \sigma', W'$ denote the idiosyncratic drift, the idiosyncratic volatility, and the idiosyncratic random driver, respectively. As $\sigma'< \sigma$, learning about the idiosyncratic drift $\mu'$ is faster than learning about the total drift $\mu$, so applying the sequential procedure in such a situation is even more advantageous than in the standard case.   

In this article, we solve the optimal liquidation problem \eqref{E:sell} within the proposed model under an arbitrary prior distribution for the drift. The first time the posterior mean of the drift passes below a specific non-decreasing curve is shown to be optimal; the stopping boundary is characterised as the unique solution of a particular integral equation.

To include more details, our investigation of the optimal strategy can be briefly described in the following. The original problem with incomplete information about the drift is reformulated as a complete information problem by projecting the price evolution onto the observable filtration using  filtering theory. 
The mean of the posterior distribution becomes the underlying process of a new equivalent optimal stopping problem with a stochastic killing/creation rate and a constant payoff function. This conditional mean is shown to satisfy a stochastic differential equation driven by the innovation process. The dispersion coefficient of the SDE is proved to be decaying in time as well as satisfy a special condition on the second spatial derivative. Embedding the value function into a Markovian framework and making a suitable connection with the term-structure equation, the established dispersion function properties enable us to employ the available convexity results to prove convexity of the Markovian value function in the spatial variable. Moreover, the value function is shown to be continuous and decreasing in time. These significant facts allow us to show that the first passage time below a monotone boundary is an optimal stopping time, so techniques from the theory of free-boundary problems with monotone boundaries can be applied. Specifically, the monotonicity of the boundary enables us to prove the smooth-fit property and to investigate the corresponding integral equation. The optimal stopping boundary is characterised as the unique non-positive and continuous solution to a non-linear integral equation.
  
Besides the examination of the optimal strategy, we investigate monotonicity properties of the expected optimal liquidation value with respect to the asset volatility and the prior distribution. Notwithstanding that all-inclusive theorems about parameter dependence appear currently to be beyond reach, we derive some sufficient conditions for monotonicity in the volatility $\sigma$ as well as the prior distribution. In addition, we conduct numerical experiments in the case of the normal prior; some results reinforce standard intuition, others illustrate inherent subtleties. In particular, additional value that an optimal strategy involving filtering brings over an optimal strategy without filtering is calculated, exhibiting an improvement of up to $10 \%$ for some feasible parameter regimes. 

As far as extensions of this work are concerned, solving the optimal liquidation problem for more general diffusions,
with possibly time- and level-dependent coefficients, is more problematic. Such extensions typically lead to the loss of the useful 
one-dimensional Markovian structure present in the classical geometric Brownian motion setting; an optimal decision then depends on the whole price trajectory rather than the current spot price alone. Clearly, a complete treatment of the resulting optimal stopping problem is much more complicated. 

\subsection{Literature review}

Over the last three decades, investment problems with incomplete information about the drift has received much attention from both financial mathematicians and financial economists.
Some of the most distinct works on portfolio optimisation include \cite{DF}, regarded as the first incomplete information problem studied in financial literature, and the general portfolio problems studied in \cite{L95, L98}; see also the recent article \cite{BDL10} proposing a general framework for most of the earlier works as well as containing an excellent survey with references. Hedging in an incomplete market under partial information about a constant drift was addressed in \cite{M} in the case of the Kalman-Bucy filter. In addition, incomplete information models have been investigated in the financial economics literature (see the survey paper \cite{B04} as well as the monograph \cite{Z03}). 

In contrast, there have been surprisingly few attempts to tackle financial optimal stopping problems under incomplete information such as the optimal liquidation problem above, with the existing works focusing mainly on a very restrictive case, namely, the two-point prior. The optimal liquidation of an asset with unknown drift has been studied in \cite{EL11} and of an asset with unknown jump intensity in \cite{L13}. For option valuation problems under incomplete information, see \cite{DMV} and \cite{G}. The financial optimal stopping articles above typically assume a two-point prior distribution; having in mind that the prior represents the beliefs about all the different values the parameter could possibly take, the two-point prior strikes as a simplistic and unrealistic assumption. Overcoming this assumption is one of the main contributions of the present article. It is also worth mentioning that various different formulations of the optimal selling problem in the case of complete information about the parameters have been studied in \cite{TP}, \cite{EHH}, and \cite{HH}.

\section{The model and problem formulation}

We consider a financial market living on a stochastic basis $(\Omega,\F, \mathbb{F}, \P)$, where the filtration $\mathbb{F}= \{\mathcal{F}_{t} \}_{t \geq 0}$ satisfies the usual conditions and the measure $\P$ denotes the physical probability measure. The basis supports a Brownian motion $W$ and a random variable $X$ such that $W$ and $X$ are independent. We assume that the observed price process $S$  evolves according to
\begin{IEEEeqnarray}{rCl} \label{E:S}
\ud S_t = X S_t \ud t +\sigma S_t \ud W_t,
\end{IEEEeqnarray}
where the volatility $\sigma>0$ is a known constant.
We write $\mathbb F^S= \lt\{ \mathcal F^S_t \rt\}_{t \geq 0}$ for the filtration generated by the price process $S$ and augmented by the null sets of $\F$. In this article,  $\mathbb F^S$ corresponds to the only available source of information, i.e.~an agent can only observe the price process $S$, but not the random driver $W$ or the drift $X$. The distribution of $X$, which we denote by $\mu$, represents the subjective beliefs of the individual about the likeliness of the different values the mean return rate $X$ may take. 
 
The optimal selling problem that we are interested in is
\begin{IEEEeqnarray}{rCl} \label{E:OS}
V=\sup_{\tau \in \mathcal{T}_T^{S} } \E[S_\tau], 
\end{IEEEeqnarray}
where $\mathcal{T}_T^{S}$ denotes the set of  $\mathbb F^S$-stopping times that are less or equal to a specified time horizon $T >0$. 

Note that if the support of $\mu$ is contained in $[0, \infty)$, then an optimal strategy is to stop at the terminal time $T$. 
Similarly, if the support of $\mu$ is contained in $(-\infty, 0]$, then an optimal strategy is to stop immediately.
To exclude these trivial cases, we from now on impose the assumption that $\mu((-\infty,0))\not= 0$ and $\mu((0,\infty))\not= 0$.

\begin{remark}
The inclusion of a constant discount rate $r>0$ is straightforward. Indeed, the discounted 
price
$\tilde S_t:=e^{-rt}S_t$ satisfies
\begin{IEEEeqnarray}{rCl}
\ud \tilde S_t = (X-r) \tilde S_t \ud t +\sigma \tilde S_t \ud W_t, 
\end{IEEEeqnarray}
and so the optimal stopping problem 
\[\sup_{\tau \in \mathcal{T}^S_T } \E[e^{-r\tau}S_\tau]\]
reduces to \eqref{E:OS} but with the prior distribution replaced by $\mu(\cdot+r)$.

%
\end{remark}

%
%

\subsection{Equivalent reformulation under a measure change}

Assuming that $\mu$ has a first moment, $\hat X_t := \E[ X \, | \, \F^S_t]$ exists,
and the process
\[ \hat W_t := \frac{1}{\sigma}\int_0^t(X - \hat X_s)  \ud s + W_t ,\] known as the innovation process, is an $\mathbb{F}^{S}$-Brownian motion (see \cite[Proposition 2.30 on p.~33]{BC09}). Writing $\mathbb F^{\hat W}=\{\F^{\hat W}_t\}_{t \geq 0}$ for the completion of the filtration \mbox{$\{ \sigma(\hat W_s : 0\leq s \leq t) \} _{t \geq 0 }$}, we note that $\mathbb F^S=\mathbb F^{\hat W}$ (see the remark on p.~35 in \cite{BC09}).

Defining a change of measure by the random variable
\[ 
\frac{\vd \Q}{\vd \P}  = e^{\sigma \hat{W}_T - \frac{\sigma^2}{2}T},
\]
and writing
\begin{IEEEeqnarray*}{rCl}
S_t &=& S_0e^{Xt+\sigma W_t-\frac{\sigma^2}{2}t}\\
&=& S_0e^{\int_0^t\hat X_s \ud s +\sigma \hat W_t-\frac{\sigma^2}{2}t},
\end{IEEEeqnarray*}
we have
\begin{IEEEeqnarray*}{rCl}
 \E \lt[ S_\tau \rt] 
= \E^{\Q} \lt[ S_0e^{ \int_0^\tau  \hat X_{s} \ud s }  \rt]= S_0 \E^{\Q} \lt[e^{ \int_0^\tau  \hat X_{s} \ud s } \rt],
\end{IEEEeqnarray*}
where $\tau \in \mathcal{T}^{S}_{T}$. Without loss of generality, we assume $S_0=1$ throughout the article; the optimal stopping problem \eqref{E:OS} then becomes
\begin{IEEEeqnarray}{rCl} \label{E:OSN}
V=\sup_{\tau \in \mathcal{T}^{S}_T } \E^{\Q}[e^{ \int_0^\tau  \hat X_{s} \ud s }].
\end{IEEEeqnarray}
We also note that, by Girsanov's theorem, the process $Z_t:= -\sigma t + \hat{W}_t$ is a $\Q$-Brownian 
motion on $[0, T]$.

\section{Analysis of the optimal stopping problem}

\subsection{Projecting onto the observable filtration}

Let us introduce $Y_t := Xt + \sigma W_t$ so that 
$S_t=S_0e^{Y_t-\frac{\sigma^2}{2}t}$. Clearly, the processes $Y$ and $S$ generate the same filtrations. 
The following proposition describes the conditional distribution of
$X$ given observations of the stock price in terms of the current value of the process $Y$. For its proof, see
Proposition 3.16 in \cite{BC09}. 

\begin{proposition}
\label{T:integrals}
Let $q:\R\to\R$ satisfy $\int_\R |q(u)| \mu(\vd u) < \infty$. Then 
\begin{IEEEeqnarray*}{rCl} 
\E\lt[ q(\Dr) \vert \mathcal F^S_t\rt] 
=\E\lt[ q(\Dr) \vert Y_t\rt]
=\frac{\int_\R q(u) e^{\frac{2uY_t -u^2t}{2\sigma^2}} \mu(\vd u)}{\int_\R e^{\frac{2uY_t -u^2t}{2\sigma^2}} \mu(\vd u)}
\end{IEEEeqnarray*} 
for any $t\geq0$.
\end{proposition}

By Proposition~\ref{T:integrals}, the distribution $\mu_{t,y}$ of $\Dr$ at time $t$ conditional on $Y_t=y$ is given by
\begin{IEEEeqnarray}{rCl}
\label{E:mu}
\mu_{t,y}(\vd u):=\frac{e^{\frac{2uy-u^2t}{2\sigma^2}} \mu(\vd u)}{\int_\R e^{\frac{2uy-u^2t}{2\sigma^2}} \mu(\vd u)},
\end{IEEEeqnarray}
and
\begin{IEEEeqnarray}{rCl} \label{E:hX}
\hat X_t = \E[ X \, | \, \F^S_t]=\E\lt[ \Dr \vert  Y_t\rt] =f(t,Y_t)
\end{IEEEeqnarray} 
for any $t>0$, where
\[f(t,y) =\int_\R u \mu_{t,y}(\vd u)=\frac{\int_\R u e^{\frac{2uy -u^2t}{2\sigma^2}} \mu(\vd u)}{\int_\R e^{\frac{2uy -u^2t}{2\sigma^2}} \mu(\vd u)}\]
As a shorthand, we denote by $\E_{t,y}$ the expectation operator under the probability measure $\P_{t,y}(\cdot ):=\P(\cdot\vert Y_t=y)$.

From now onwards, the following integrability condition on $\mu$ is imposed.
\begin{assump}
\label{assump}
The prior distribution $\mu$ satisfies 
\begin{IEEEeqnarray}{rCl}
\label{integrability}
 \int_\R e^{a u^2}\mu(\vd u)<\infty 
\end{IEEEeqnarray}
for some $a>0$.
\end{assump}
This assumption is an insignificant restriction on our optimal liquidation problem, since, given any probability distribution $\mu$, 
the distributions $\mu_{t,y}$ in \eqref{E:mu} satisfy \eqref{integrability} for any $t>0$. 
The main benefit of Assumption~\ref{assump} is that it allows us to extend the definition of $\mu_{t,y}$ in \eqref{E:mu} to $t=0$.
Indeed, suppose that $\mu$ satisfies \eqref{integrability} with $a=\ep/(2\sigma^2)$ for some $\ep>0$ . Defining a probability distribution $\xi$ on $\R$ by
\begin{equation} \label{E:EP}	
	\xi(\vd u) := \frac{e^{\frac{\ep u^2}{2\sigma^2}}\mu(\vd u)}{\int_\R e^{\frac{\ep u^2}{2\sigma^2}}\mu(\vd u)},
\end{equation}
the measure
\begin{IEEEeqnarray*}{rCl} 
\mu_{0,y} (\vd u) = \frac{e^{\frac{uy}{\sigma^2}} \mu(\vd u)}{\int_\R e^{\frac{uy}{\sigma^{2}}} \mu(\vd u) }
\end{IEEEeqnarray*}
coincides with
\begin{IEEEeqnarray*}{rCl} 
\xi_{\ep, y}(\vd u) := \frac{e^{\frac{2uy-u^2 \ep}{2\sigma^2}} \xi(\vd u)}{\int_\R e^{\frac{2uy-u^2 \ep}{2\sigma^2}} \xi(\vd u) }.
\end{IEEEeqnarray*}
Consequently, the distribution $\mu_{0,y}$ can be identified with a conditional distribution at time $0$ given that the prior distribution 
at time $-\ep$ was $\xi$ and the current value of the observation process is $y$.
This gives us a generalisation of the notion of the starting point of the observation process $Y$ to allow \mbox{$Y_0=y \neq 0$},
so we may regard time 0 as an interior point of the time interval.

Next, we establish a bijective correspondence between $S_{t}$ and $\hat X_{t}$. For it, let $I_\mu$ denote the interior of the smallest closed interval containing the support of $\mu$, i.e.~$I_\mu = (\inf(\supp (\mu)), \sup(\supp (\mu)))$.

\begin{lemma}
\label{f-1-1}
For any given $t \geq 0$, the function $f(t,\cdot): \R \to I_\mu$ defined  above is a strictly increasing continuous bijection.
\end{lemma}


\begin{proof}
Thanks to Assumption \ref{assump}, it suffices to prove the claim only for $t=0$. 
Differentiation of $f$ under the integral sign yields $\pt_2 f(0,y) = \frac{1}{\sigma^{2}}(\E_{0,y}[X^2]-\E_{0,y}[X]^2)$, which is 
strictly positive and finite for all $y \in \R$. As a result, $y \mapsto f(0, y)$ is strictly increasing. 
For surjectivity, we need that $f(0, y) \tend \sup I_\mu$ as $y \tend \infty$ and $f(0, y) \tend \inf I_\mu$ as $y \tend -\infty$. We only prove the first claim as the second one then follows immediately by symmetry. 

Let $\theta \in I_\mu \cap (0, \infty)$, $y >0$, and consider
\begin{IEEEeqnarray}{rCl}
\int_\R u e^{\frac{uy}{\sigma^2}} \mu( \vd u) -\theta \int_\R e^{\frac{uy}{\sigma^2}} \mu(\vd u) 
&=& \int_\R (u- \theta) e^{\frac{uy}{\sigma^2}} \mu(\vd u) 
= e^{\frac{\theta y}{\sigma^2}} \int_\R w e^{\frac{wy}{\sigma^2}} \mu( \theta + \vd w), \label{E:diff}
\end{IEEEeqnarray}
where $w :=u-\theta$.
As the minimum of $w \mapsto w e^{\frac{wy}{\sigma^2}}$ is attained at $w = -\sigma^2/y$, we have
\[ \int_{(-\infty, 0]} w e^{\frac{wy}{\sigma^2}} \mu(\theta + \vd w)  \geq -\frac{\sigma^2 e^{-1}}{y} \int_{(-\infty, 0]} \mu(\theta + \vd w) \geq - \frac{\sigma^2}{y}. \]
Furthermore, 
\[ \int_{(0, \infty)} w e^{\frac{wy}{\sigma^2}} \mu (\theta + \vd w)  \tend \infty \]
as $y \tend \infty$ by monotone convergence. 
Consequently, from \eqref{E:diff} follows that $f(0, y) \geq \theta $ for all large enough $y$. Since $\theta \in I_\mu $ was arbitrary, we conclude that $f(0,y) \tend \sup I_\mu$ as $y \tend \infty$, which finishes the proof.
\end{proof}
Writing $\mathbb{F}^{\hat X}=\{\F^{\hat X}_{t}\}_{t \geq 0}$ for the completion of the filtration generated by $\hat X$ and writing $\mathcal{T}_{T}^{\hat X}$ for the set of $\mathbb{F}^{\hat X}$-stopping times not exceeding $T$, we formulate the following immediate corollary.
\begin{corollary} 
$\mathbb{F}^{S}=\mathbb{F}^{\hat X}$ and $\mathcal{T}_{T}^{S}= \mathcal{T}_{T}^{\hat X}$.
\end{corollary}
\noindent A consequence of this corollary is that the optimal stopping problem \eqref{E:OSN} can be rewritten as
\begin{IEEEeqnarray}{rCl} \label{E:OSNX}
V=\sup_{\tau \in \mathcal{T}^{\hat X}_T } \E^{\Q}[e^{ \int_0^\tau  \hat X_{s} \ud s }].
\end{IEEEeqnarray}
Looking for a more tractable characterisation of $\hat X$, we find an SDE representation of $\hat X$ with respect to the observations filtration $\mathbb{F}^{S}$.
An application of It\^o's formula to $\hat X_t=f(t, Y_t)$ yields
\begin{IEEEeqnarray}{rCl} \label{E:SDEB}
\ud \hat X_t =\sigma \pt_2 f(t, Y_t) \ud \hat{W}_t.
\end{IEEEeqnarray}
Introducing the notation $f_t := f(t, \cdot)$, we define
\[ \psi(t, x) := \sigma\pt_2 f(t, f_t^{-1}(x)) \]
and, from \eqref{E:SDEB} above, obtain a stochastic differential equation
\begin{IEEEeqnarray*}{rCl}
\ud \hat X_{t} &=& \psi( t, \hat X_t) \ud \hat{W}_t
\end{IEEEeqnarray*}
for the conditional mean $\hat X_{t}$. 
Rewriting of the equation in terms of the $\mathbb Q$-Brownian motion $Z_t= -\sigma t + \hat{W}_t$ results in
\begin{IEEEeqnarray}{rCl}
\label{X}
\vd \hat X_t &=& \sigma \psi(t, \hat X_t) \ud t + \psi(t, \hat X_t) \ud Z_t.
\end{IEEEeqnarray}
The dispersion $\psi$ can be expressed more explicitly (by differentiating $f$ under the integral sign) as
\[ \psi(t,x) = \frac{1}{\sigma} \lt( \E_{t,y_{x}(t)}[X^{2}] - \E_{t, y_{x}(t)}[X]^{2} \rt) = \frac{1}{\sigma} \Var_{t,y_{x}(t)}(X), \]
where the notation $y_x(t):=f_t^{-1}(x)$ is used (note that $y_x(t)$ is the unique value of the observation process $Y_t$
that yields $\hat X_t=x$).

\begin{example}{\bf (The two-point prior)} Suppose $\mu=\pi \delta_h + (1-\pi) \delta_l$, where $\delta_l, \delta_h$ denote {the Dirac measures} at $l, h\in \R$ respectively. Then $ \psi(t,x) =\frac{1}{\sigma} (h-x)(x-l)$. 
\end{example}\label{Ex:B}

\begin{example}{\bf (The normal prior)} Suppose $\mu$ is the normal distribution with mean $m$ and variance $\gm^2$. Then the conditional distribution $\P(\cdot\vert Y_t=y)=\mu_{t,y}$ is also normal but with mean $\frac{\sigma^2 m +\gamma^2y}{\sigma^2+t\gamma^2}$
and variance $\frac{\sigma^2\gamma^2}{\sigma^2+t\gamma^2}$. Consequently, $\psi(t,x) = \frac{\sigma\gamma^2}{\sigma^2+t\gamma^2}$. \label{Ex:N}
\end{example}

\subsection{Dispersion of the conditional mean}

The following inequality will be the key to understanding the dispersion function $\psi$.
\begin{proposition} \label{E:HighDimIneq}
Let $X$ be a random variable with $\E[X^4]<\infty$. Then 
\[ 
\E [X^4]\E [X^2] + 2\E[X^3]\E[X^2]\E[X] - \E[X^4] \E [X]^2 - \E[X^3]^2 - \E[X^2]^3 \geq 0
\] 
with the equality if and only if $X$ has a one-point or a two-point distribution.

\end{proposition}
\begin{proof}
Let $X, Y, Z$ be independent and identically distributed random variables with $\E[X^4] < \infty$. Observe that
\begin{IEEEeqnarray*}{rCl}
\IEEEeqnarraymulticol{3}{l}{
\E[(X-Y)^2(Y-Z)^2(Z-X)^2]} \\
&=& \E[X^4(Y^2+Z^2)+Y^4(Z^4+X^4) +Z^4(X^4+Y^4)] \\
&&-2 \E[X^4YZ +Y^4ZX +Z^4XY] \\
&&+2 \E[X^3(Y^2Z+Z^2Y) + Y^3(Z^2X+X^2Z) +Z^3(X^2Y + Y^2X)] \\
&& -2 \E[X^3Z^3+Y^3X^3 +Z^3Y^3]-6\E[X^2Y^2Z^2] \\
&=& 6(\E[X^4]\E [X^2]  - \E[X^4] \E [X]^2 + 2\E[X^3]\E[X^2]\E[X]- \E[X^3]^2 - \E[X^2]^3),
\end{IEEEeqnarray*}
where the last equality holds because $X,Y, Z$ are i.i.d. It is clear that
\[ \E[(X-Y)^2(Y-Z)^2(Z-X)^2] \geq 0 \]
with the equality if and only if $X$ has a one-point or a two-point distribution. This finishes the proof of the claim.
\end{proof}
\begin{remark}We are grateful to Johan Tysk for providing an alternative proof of the above proposition based on the Pythagorean theorem in $L^2$ spaces.
\end{remark}

\begin{proposition}[Properties of the dispersion function $\psi$] \item \label{T:psigr}
\begin{enumerate} 
\item
For any $x \in I_\mu$, the function $t \mapsto \psi(t, x)$ is non-increasing. 
It is strictly decreasing unless $\mu$ is a two-point distribution, in which case $t \mapsto \psi(t,x)$ is a constant. 
\item
$\pt_2^2 \psi \geq -\frac{2}{\sigma}$ with a strict inequality unless $\mu$ is a two-point distribution, in which case we have equality. 
\item If $\mu$ is compactly supported, then $\psi$ is bounded.
\end{enumerate} 
\end{proposition}

\begin{proof}
\begin{enumerate}
\item
Recall the notation $y_x(t)=f_t^{-1}(x)$, and consider
\begin{IEEEeqnarray*}{rCl}
\pt_{1} \psi(t,x) &=& \frac{\pt}{\pt t} \lt( \frac{1}{\sigma} \lt( \E_{t, y_{x}(t)}[X^{2}] - \E_{t, y_{x}(t)}[X]^{2} \rt) \rt) \\
&=& \frac{\pt}{\pt t} \lt( \frac{1}{\sigma} \lt( \E_{t, y_{x}(t)}[X^{2}] -x^{2} \rt) \rt) \\
&=& \frac{1}{\sigma} \Bigg( \E_{t, y_{x}(t)} \lt[X^{2} \lt( \frac{y_{x}'(t)}{\sigma^{2}}X- \frac{1}{2 \sigma^{2}} X^{2} \rt) \rt] \\
&&- \E_{t, y_{x}(t)} \lt[X^{2} \rt] \E_{t, y_{x}(t)} \lt[ \frac{y_{x}'(t)}{\sigma^{2}}X- \frac{1}{2 \sigma^{2}} X^{2} \rt] \Bigg) \\
&=& \frac{1}{\sigma^{3}} \Big( y_{x}'(t) \lt( \E_{t, y_{x}(t)}[X^{3}] - \E_{t, y_{x}(t)}[X^{2}] \E_{t, y_{x}(t)}[X] \rt) \\
&& -\frac{1}{2} \lt( \E_{t, y_{x}(t)}[X^{4}] - \E_{t, y_{x}(t)}[X^{2}]^{2}\rt) \Big)
\end{IEEEeqnarray*} 
Implicit differentiation using the identity $x = f(t, y_{x}(t))$ gives that
\begin{IEEEeqnarray*}{rCl}
y_{x}'(t) &=& \frac{1}{2} \frac{\E_{t, y_{x}(t)}[X^{3}] - \E_{t, y_{x}(t)}[X^{2}]\E_{t, y_{x}(t)}[X]}{\Var_{t, y_{x}(t)}(X)},
\end{IEEEeqnarray*}
which substituted into the last expression above yields 
\begin{IEEEeqnarray*}{rCl}
\pt_{1} \psi(t,x) &=& \frac{ \lt( \E_{t, y_{x}(t)}[X^{3}] - \E_{t, y_{x}(t)}[X^{2}]\E_{t, y_{x}(t)}[X] \rt)^{2} - \Var_{t, y_{x}(t)} (X^{2}) \Var_{t, y_{x}(t)}(X)}{2 \sigma^{3} \Var_{t, y_{x}(t)}(X)} \\\
&=& \frac{-1} {2 \sigma^{3} \Var_{t, y_{x}(t)}(X)} \bigg( \E_{t, y_{x}(t)} [X^4]\E_{t, y_{x}(t)} [X^2] \\
&&+ 2\E_{t, y_{x}(t)}[X^3]\E_{t, y_{x}(t)}[X^2]\E_{t, y_{x}(t)}[X] \\
&&- \E_{t, y_{x}(t)}[X^4] \E_{t, y_{x}(t)} [X]^2 - \E_{t, y_{x}(t)}[X^3]^2 - \E_{t, y_{x}(t)}[X^2]^3 \bigg) .
\end{IEEEeqnarray*}

Now, the claim follows from Proposition \ref{E:HighDimIneq} applied to the term between the parentheses.
\item
By the chain rule applied to the definition of $\psi$, we have 
\[\pt_2 \psi(t, x) = \sigma \pt_2^2 f (t, y_x(t)) \pt_2 y(t, x), \]
where $y(t,x):=y_x(t)$. Here
\begin{IEEEeqnarray*}{rCl}
\pt_{2}^{2} f (t, y) &=& \frac{1}{\sigma^{2}} \frac{\pt }{\pt y} \lt( \E_{t,y}[X^{2}]- \E_{t,y}[X]^{2} \rt) \\
&=& \frac{1}{\sigma^{4}} \lt( \E_{t,y}[X^{3}] -3 \E_{t,y}[X^{2}]\E_{t,y}[X] + 2 \E_{t,y}[X]^{3} \rt)
\end{IEEEeqnarray*}
by straightforward differentiation under the integral sign, and
\[ \pt_2 y(t, x) = \frac{1}{\pt_2 f(t, y(t, x))} = \frac{\sigma^{2}}{\Var_{t, y_x(t)} (X)} .\]
by implicit differentiation. Hence
\begin{IEEEeqnarray*}{rCl}
\pt_{2} \psi(t,x) &=& \frac{1}{\sigma} \lt( \frac{\E_{t, y_{x}(t)}[X^{3}] - \E_{t, y_{x}(t)}[X^{2}]\E_{t, y_{x}(t)}[X]}{\Var_{t, y_{x}(t)}(X)} - 2 x \rt). 
\end{IEEEeqnarray*}
It remains to establish the inequality
\begin{IEEEeqnarray*}{rCl}
\pt_2^2 \psi (t,x) +\frac{2}{\sigma} 
&=& \frac{1}{\sigma}\frac{\pt}{\pt y} \lt(\frac{\E_{t, y}{[X^3]} - \E_{t, y}[X^2] \E_{t, y}[X]}{\Var_{t, y}(X)} \rt)\Bigg|_{y=y_x(t)}  \pt_2 y(t, x) \geq 0. 
\end{IEEEeqnarray*}

As $\pt_2 y > 0$, equivalently, it suffices to prove the non-negativity of 
\begin{IEEEeqnarray*}{rCl}
q(t,y) &:=& \Var_{t,y}(X)^{2} \frac{\pt}{\pt y} \lt(\frac{\E_{t, y}[X^3] - \E_{t, y}[X^2] \E_{t, y}[X]}{\Var_{t,y}(X)}\rt)
\\
&=&  \frac{\pt}{\pt y} \lt( \E_{t,y}[ X^3] \rt) \Var_{t,y} (X) - \E_{t,y}[X^3] \frac{\pt}{\pt y} \lt( \Var_{t,y} X \rt) \\
&&- \frac{\pt}{\pt y} \lt( \E_{t,y} [X^2] \E_{t,y}[X] \rt) \Var_{t,y} (X) + \E_{t,y}[X^2]\E_{t,y}[X] \frac{\pt}{\pt y} \Var_{t,y} (X).
\end{IEEEeqnarray*}
Further differentiation yields that
\begin{IEEEeqnarray*}{rCl}
q(t, y) &=&\frac{1}{\sigma^{2}}\bigg(\E_{t,y}[ X^4]\E_{t,y} [X^2] +2 \E_{t,y} [X^3] \E_{t,y} [X^2] \E_{t,y} [X] \\
&& - \E_{t,y} [X^4] \E_{t,y}[ X]^2 - \E_{t,y}[ X^3]^2 - \E_{t,y}[ X^2]^3\bigg).
\end{IEEEeqnarray*}
Thus, by Proposition \ref{E:HighDimIneq}, $q \geq 0$; moreover, $q>0$ for all priors $\mu$ except the two-point distribution in which case $q=0$.
\item
The well-known identity 
\[ \E_{t, y_{x}(t)}[|X|^2]=2 \int_{[0, \infty)} u \P_{t, y_{x}(t)}(| X |> u) \ud u \]
ensures that $\psi$ is bounded for compactly supported distributions. 
\end{enumerate}
\end{proof}
\begin{remark} \item
\begin{enumerate}
\item
It is possible to come up with a contrived example of a prior distribution for which the dispersion $\psi$ is unbounded. For this, think of a discrete probability measure supported on an infinite number of points  $x_1 < x_2 < \ldots < x_n < \ldots$ such that $x_{n}-x_{n-1} \tend \infty$ as $n \tend \infty$. Using the notation $\bar x_n := (x_{n-1}+ x_n)/2$ for the mean between neighbouring points, the value $\psi(t,\bar x_{n})=Var_{t,\bar x_n}(X)/\sigma \tend \infty$ as $n\tend \infty$ by comparison with a two-point distribution concentrated at the points $x_{n-1}$ and $x_{n}$.
\item
Let us stress that compact support of the prior is by no means a necessary condition for the boundedness of $\psi$. For instance, we know that $\psi$ is bounded in the case of a normal prior as seen in the example on page \pageref{Ex:N}. Though a rigorous investigation into precise technical conditions on the prior for the boundedness of $\psi$ appears to be involved enough to be omitted in this article, we conjecture, based on numerical investigations, that $\psi$ is bounded for any prior admitting a density that monotonically approaches zero outside a large enough finite-length interval around the origin. 
\end{enumerate}
\end{remark}
As the boundedness of $\psi$ appears to be satisfied by any conceivable prior of interest in practical applications, we make it an assumption in the rest of the article.
\begin{assump}
\label{assumpbdd}
The prior distribution $\mu$ is such that $\Var_{0,y}(X) < \infty $ for all $y \in \R$.
\end{assump}



\subsection{The Markovian value function and the optimal strategy}

Using the dynamics \eqref{X} of $\hat X$, we are able to embed the optimal stopping problem \eqref{E:OSNX} into a Markovian framework.
To do that, define
\begin{IEEEeqnarray}{rCl} \label{E:VM}
v(t,x) 
&=& \sup_{\tau \in \mathcal{T}_{T-t}} \E^{\Q} \lt[ e^{ \int_0^\tau  \hat X^{t,x}_{t+s} \ud s } \rt] \label{E:OSP}, \quad (t,x) \in [0, T]\times I_{\mu},
\end{IEEEeqnarray}
where the process $\hat X= \hat X^{t,x}$ is given by
\begin{IEEEeqnarray*}{rCl} \label{E:Bhm}
\lt\{ \begin{array}{ll}
\vd \hat X_{t+s} = \sigma \psi(t+s, \hat X_{t+s}) \ud s + \psi(t+s, \hat X_{t+s}) \ud Z_{t+s} & \quad (s >0), \\
\hat X_t = x,\end{array} \rt.
\end{IEEEeqnarray*} 
and $\mathcal T_{T-t}$ denotes the set of stopping times less or equal to $T-t$ with respect to the completed filtration of $\{ \hat X^{t,x}_{t+s}\}_{s \geq 0}$.


Let us define the sets
\[\mathcal C =\{(t,x)\in[0,T] \times I_\mu :v(t,x)> 1 \}\]
and
\[\mathcal D=\{(t,x)\in[0,T] \times I_\mu :v(t, x)=1\},\]
which we will soon show  to correspond respectively to continuation and stopping sets of an optimal strategy. 
Note that 
\begin{IEEEeqnarray}{rCl}
\label{v1}
v\geq 1
\end{IEEEeqnarray}
everywhere and that $v(T,x)=1$, so $\mathcal{C} \cup \mathcal D = [0,T]\times I_\mu$ and the random time
\begin{IEEEeqnarray}{rCl}
\label{tauD}
\tau_\mathcal{D} := \inf \{ s \geq 0: (t+s,\hat X^{t,x}_{t+s})\in\mathcal D\}
\end{IEEEeqnarray}
satisfies $\tau_\mathcal{D}\leq T-t$.

\begin{proposition}[Optimal stopping time] \label{T:OST}
The value function $v$ is finite, and
the time $\tau_\mathcal{D}$ defined in \eqref{tauD} is an optimal stopping time. 
\end{proposition} 

\begin{proof}
Without loss of generality, assume that $t=0$, and let $x\in\R$.
By Theorem D.12 in \cite{KS2}, to prove the claims it suffices to show that
\[
\E^{\Q} \bigg[\sup_{0\leq t \leq T} \exp \bigg({{\int_0^t  \hat X^{0,x}_s \ud s}}\bigg) \bigg] < \infty.
\]
By the Dambis-Dubins-Schwartz theorem, there exists (possibly on a larger probability space) a Brownian motion $B$ such that
\[\int_0^t \psi(s, \hat X^{0,x}_{s}) \ud Z_{s}=B_{\int_0^t\psi(t, \hat X^{0,x}_s)^2\ud s}\,.\]
If $m >0$ is a constant dominating $\psi$, then
\begin{IEEEeqnarray*}{rCl}
\sup_{0\leq t \leq T} \exp \bigg({{\int_0^t  \hat X^{0,x}_s \ud s}}\bigg) &\leq& \exp \bigg(  T \sup_{0 \leq t \leq T} \hat X^{0,x}_t \bigg) \\
&\leq& \exp \bigg(  T \bigg( x +\sigma m T + \sup_{0\leq t \leq T} B_{\int_0^t \psi(s, \hat X^{0,x}_s)^2 \ud s} \bigg)  \bigg) \\
&\leq& \exp \bigg(  T \bigg( x +\sigma m T + \sup_{0\leq t \leq m^2T} B_{t} \bigg)  \bigg). 
\end{IEEEeqnarray*}
Thus
\begin{IEEEeqnarray*}{rCl}
\E^{\Q} \bigg[\sup_{0\leq t \leq T} \exp \bigg({{\int_0^T  \hat X^{0,x}_s \ud s}}\bigg) \bigg] &\leq& \E^\Q \bigg[ \exp \big(  T \big( x +\sigma m T + \sup_{0\leq t \leq m^2T} B_{t} \big)  \big) \bigg] \\
&=& \exp \lt(  T \lt( x +\sigma m T \rt) \rt) \E^\Q \lt[ \exp \lt( T  |B_{m^2T}|  \rt) \rt] \\
&<& \infty, 
\end{IEEEeqnarray*}
where the equality comes from the reflection principle. 
\end{proof}

Since $\psi$ is continuously differentiable, it is Lipschitz continuous on any compact subset of $[0,T]\times I_{\mu}$. 
To avoid additional technical complications, from now on we impose a slightly stronger assumption of Lipschitz continuity on the whole of $[0,T]\times I_{\mu}$

\begin{assump}
\label{assumpLip}
The function $\psi$ is Lipschitz continuous in the second variable, i.e. there exists $K >0$   such that $|\psi(t,x)-\psi(t,y)| \leq K |x-y|$ for all $t\in [0, T]$ and all $x, y \in I_{\mu}$.
\end{assump}
\noindent We remark that our canonical examples of the normal and the two-point prior both fulfill Assumption~\ref{assumpLip}.

\begin{theorem}[Properties of the value function] \label{T:VF}
\item
\begin{enumerate}
\item
The function  $x \mapsto v(t, x)$ is non-decreasing and convex for any fixed $t \in [0, T]$. 
\item
The function $t \mapsto v(t,x)$ is non-increasing for any fixed $x \in I_\mu$. 
\item
The value function $v$ is continuous on $[0,T]\times I_\mu$. 
\item
There exists a non-decreasing and continuous function $h : [0, T] \to (-\infty,0]$ with $h(T)=0$
such that $\mathcal C=\{(t,x) \in [0, T)\times I_\mu :x > h(t)\}$.  
\item
The value function $(t, x) \mapsto v(t,x)$ solves the boundary value problem 
\begin{IEEEeqnarray}{rCl}
\begin{cases} \label{E:PDEc}
\pt_1 v+\sigma \psi(t, x) \pt_2 v + \frac{1}{2} \psi(t,x)^2 \pt_2^2 v+ x v =  0, & \; x \in \mathcal C, \label{E:PDE}\\
v= 1, & \; x \in \mathcal{D}.
\end{cases}
\end{IEEEeqnarray}
Furthermore, the smooth-fit property holds in that the function $x \mapsto v(t,x)$ is $C^{1}$ for all $t \in [0, T]$.
\end{enumerate}
\end{theorem}

\begin{proof}\item
\begin{enumerate}
\item
\begin{enumerate}[(i)]
\item
The monotonicity of $x\mapsto v(t,x)$ is clear from the representation
\begin{IEEEeqnarray}{rCl} \label{E:vm}
v(t,x) =\sup_{\tau \in \mathcal{T}_{T-t}} \E^{\Q} \lt[ e^{ \int_0^\tau  \hat X^{t,x}_{t+s} \ud s }\rt]
\end{IEEEeqnarray} 
of the value function
together with a comparison theorem, see \cite[Theorem~IX.3.7]{RY}.
\item 
Let us define  $v_{E}(t,x):= \E^{\Q} \lt[ e^{ \int_0^{T-t}  \hat X^{t,x}_{t+s} \ud s }\rt]$ and $u_{E}(t,r):=\E^{\Q} \lt[ e^{ -\int_0^{T-t} \hat R^{t,r}_{t+s} \ud s }\rt]$, where $\hat R= -\hat X$ and so
\[ \vd \hat R_{t}= -\sigma \psi(t, -\hat R_{t}) \ud t - \psi(t, -\hat R_{t}) \ud Z_{t}.\]
Then $v_{E}(t,x) = u_{E}(t, - x)$. Now, the convexity result follows by approximating the value function, starting with $v_{E}$ as the first approximation, by Bermudan options, which preserve convexity by \cite[Theorem 5.1]{ET08} with the needed condition $\pt_2^2 \psi \geq -\frac{2}{\sigma}$
for the theorem to hold ensured by Proposition \ref{T:psigr}.
\end{enumerate}
\item
From Proposition \ref{T:psigr}, the dispersion $\psi$ is non-increasing in $t$, so the claim follows by the Bermudan approximation argument for the value function and Theorem 6.1 in \cite{ET08}.
\item 
First, let $l \in I_{\mu}$ and we will show that there exists a constant $K>0$ such that, for every $t \in [0, T]$, the map $x \mapsto v(t,x)$ is $K$-Lipschitz continuous on $(-\infty, l] \cap I_{\mu}$. Assume for a contradiction that there is no such $K$. Recall that convexity of a single-variable function implies continuity and existence of one-sided derivatives. Hence using a characterisation of convexity saying that a real-valued function $f$ defined on an interval is convex if and only if the function $(x_{1}, x_{2}) \mapsto (f(x_{2})-f(x_{1}))/(x_{2}-x_{1})$ is increasing in both $x_{1}$ and $x_{2}$, we obtain that there is a sequence $\{t_{n}\}_{n\geq0} \subset [0,T]$ such that the sequence of left-derivatives $\pt^{-}_{2}v(t_{n},l)$ diverges to $\infty$. However, taking $\epsilon \in (0, \sup I_{\mu} -l )$, this would imply that $v(t_{n}, l+\epsilon) \tend \infty$, contradicting the fact that $v(t_{n}, l+ \epsilon) \leq v(0, l+ \epsilon) < \infty $ for all $n \in \N$. 

To finish the proof of the continuity of $v$, it suffices to show that $v(t,x)$ is continuous in $t$. To reach a contradiction, 
assume that $t\mapsto v(t,x_0)$ is not continuous at $t=t_0$ for some $x_0$. By time-decay, this means that $v$ has a negative jump.

First consider the case when $v(t_0-,x_0)>v(t_0,x_0)$. By Lipschitz continuity in the second variable, there exists a rectangle
$\mathcal R =(t_0-\delta,t_0)\times (x_0-\delta,x_0+\delta)$ with $\delta>0$ such that
\begin{IEEEeqnarray}{rCl}
\label{johnlennon}
\inf_{(t,x)\in\mathcal R}v(t,x)>v(t_0,x_0+\delta).
\end{IEEEeqnarray}
Thus $\mathcal R\subseteq \mathcal C$.
Let $t\in(t_0-\delta,t_0)$ and $\tau_\mathcal R:=\inf\{u\geq 0: (t+u, \hat X^{t,x_0}_{t+u})\notin\mathcal R\}$.
Then, by martingality in the continuation region (see \cite[Appendix D]{KS2}),
\begin{IEEEeqnarray*}{rCl}
v(t,x_0) &=& \E^{\Q} \lt[ e^{ \int_0^{\tau_{\mathcal R}} \hat X^{t,x_0}_{t+u} \ud u }v(t+\tau_\mathcal R, \hat X^{t,x_0}_{t+\tau_{\mathcal R}})\rt]\\
&\leq& \E^{\Q} \lt[ e^{ \int_0^{t_0-t}\hat X^{t,x_0}_{t+u}\vee 0 \ud u }v(t,x_0+\delta)\Ind_{\{t+\tau_\mathcal R <t_0\}}\rt]\\
&&+\E^{\Q} \lt[e^{ \int_0^{t_0-t}\hat X^{t,x_0}_{t+u}\vee 0 \ud u }v(t_0,x_0+\delta)\Ind_{\{t+\tau_\mathcal R = t_0\}}\rt]\\
&\leq& e^{(t_0-t)(x_0+\delta)^+}v(t,x_0+\delta) \Q(t+\tau_\mathcal R <t_0) +e^{(t_0-t)(x_0+\delta)^+}v(t_0,x_0+\delta)\\
&\to& v(t_0,x_0+\delta)
\end{IEEEeqnarray*}
as $t\to t_0$, which contradicts \eqref{johnlennon}.

Next, consider the case when $v(t_0,x_0)>v(t_0+,x_0)$. We begin by investigating the situation $v(t_0,x_0)>v(t_{0}+, x_{0})>1$. By Lipschitz continuity of $v$ in the second variable and its decay in time, there exists $\mathcal R=(t_0,t_0+\epsilon] \times [x_0-\delta,x_0+\delta]$ with $\epsilon>0$ and $\delta >0$ such that 
\begin{IEEEeqnarray}{rCl}
\label{paul}
v(t_0,x_0)>\sup_{(t,x) \in \mathcal R }v(t,x) \geq \inf_{(t,x) \in \mathcal R }v(t,x)>1.
\end{IEEEeqnarray}
In particular, $\mathcal R \subseteq \mathcal{C}$ and writing $\tau_\mathcal R:=\inf\{u\geq 0: (t_{0}+u, \hat X^{t_{0},x_0}_{t_{0}+u})\notin\mathcal R\}$ we have
\begin{IEEEeqnarray*}{rCl}
v(t_{0},x_0) &=& \E^{\Q} \lt[ e^{ \int_{0}^{\tau_{\mathcal R}} \hat X^{t_{0},x_0}_{t_{0}+u} \ud u }
v(t_0+\tau_\mathcal R, \hat X^{t_{0},x_0}_{t_0+\tau_{\mathcal R}})\rt] \\
&\leq& \E^{\Q} \lt[ e^{ \int_0^{\epsilon}\hat X^{t_{0},x_0}_{t_{0}+u}\vee 0 \ud u }v(t_{0},x_0+\delta)\Ind_{\{\tau_\mathcal R <\epsilon\}}\rt]\\
&&+\E^{\Q} \lt[e^{ \int_0^{\epsilon}\hat X^{t_{0},x_0}_{t_0+u}\vee 0 \ud u }v(t_0+\epsilon,x_0+\delta)\Ind_{\{\tau_\mathcal R = \epsilon\}}\rt]\\
&\leq& e^{\epsilon(x_0+\delta)^+}v(t_{0},x_0+\delta) \Q(\tau_\mathcal R <\epsilon) +e^{\epsilon(x_0+\delta)^+}v(t_0+\epsilon,x_0+\delta)\\
&\to& v(t_0+,x_0+\delta)
\end{IEEEeqnarray*}
as $\epsilon \searrow 0$, which contradicts \eqref{paul}.

Alternatively, suppose that $v(t_0,x_0)>v(t_{0}+, x_{0})=1$. By Lipschitz continuity in the second variable, there exists $\delta>0$ such that 
\begin{IEEEeqnarray}{rCl}
\label{ringo}
\inf_{x\in(x_{0}-\delta, x_{0})}v(t_{0},x)>v(t_0+,x_0)=1.
\end{IEEEeqnarray}
Then $(t_{0}, T]\times (x_{0}-\delta, x_{0}) \subseteq \mathcal{D}$ and so the process $\hat X^{t_{0}, x_{0}-\delta/2}$ hits the stopping region immediately, implying that $(t_{0}, x_{0}-\delta/2) \in \mathcal{D}$; this contradicts $\eqref{ringo}$.

\item 

Existence of a non-decreasing boundary $h : [0,T) \into [-\infty, \infty]$ satisfying $ \mathcal C=\{(t,x) \in [0, T)\times I_\mu :x > h(t)\} $ is a direct consequence of the first two parts above, and we can then define $h(T)=\lim_{t \nearrow T}h(t)$. 
Non-positivity of $h$ is clear from the expression \eqref{E:OSP}, since, for any starting point $(t,x)\in  [0, T)\times (0, \infty)$, the strategy of stopping at the first time $\hat X^{t,x}$ hits $0$ 
gives a value strictly greater than $1$. 

To show that $h$ is bounded from below, assume for a contradiction that $\{0\} \times (-\infty, \infty) \subseteq \mathcal{C}$. Hence, defining $\xi$ as in \eqref{E:EP}, we know that $(-\epsilon,0] \times \R \subseteq \mathcal{C}_{\xi}$, where $\mathcal{C}_{\xi}$ denotes the continuation region for the optimal selling problem started at time $-\epsilon<0$ with the prior $\xi$. 
Writing $v_{\xi}$ to denote the Markovian value function for the selling problem from time $-\epsilon$,
let $-t'\in(-\epsilon,0)$ and let $a<0$ be such that $v_\xi(-t',a)<e^{-at'}$.
Now, let $x \in (-\infty, a)$, and observe that
\begin{IEEEeqnarray*}{rCl}
v_{\xi}(-t',x) &\leq& e^{at'} v_\xi(-t',a)\P \bigg(\sup_{0\leq u \leq t'} \hat X^{-t',x}_{-t'+u} < a\bigg) + v_{\xi}(-t',a) \P \bigg(\sup_{0\leq u \leq t'} \hat X^{-t',x}_{-t'+u} \geq a \bigg) \\
&\tend& e^{at'}v_\xi(-t',a)<1
\end{IEEEeqnarray*}
as $x\searrow -\infty$. This gives a contradiction since $v_{\xi} \geq 1$ by definition. As a result, we can conclude that $h(t) \in (-\infty, 0]$ for all 
$t\in [0,T]$.

For the continuity of $h$, note that continuity together with time-decay of $v$ imply that $h$ is right-continuous with left limits.
Assume for a contradiction that $h(t_0-)<h(t_0)$ for some $t_0\in(0,T)$. Take points $x_1,x_2$ with $h(t_0-)<x_1<x_2<h(t_0)$, 
let $x=(x_1+x_2)/2$, and
consider the rectangle $\mathcal R=(t_0-\delta,t_0)\times(x_1,x_2)\subseteq \mathcal C$ for some $\delta>0$.
For $t\in(t_0-\delta,t_0)$, 
\begin{IEEEeqnarray}{rCl}
\label{E:cont}
v(t,x) &=& \E^{\Q} \lt[ e^{ \int_0^{\tau_{\mathcal R}} \hat X^{t,x}_{t+u} \ud u }
v(t+\tau_\mathcal R, \hat X^{t,x}_{t+\tau_{\mathcal R}})\rt] \\
\notag
&\leq& \Q\left(\tau_{\mathcal R}<t_0-t\right)v(t_0-\delta,x_2) + e^{x_2(t_0-t)}.
\end{IEEEeqnarray}
Now, estimating $\tau_{\mathcal R}$ above with the leaving time of $\mathcal R$ for a Brownian motion
with drift (compare the proof of Proposition~\ref{T:OST}), it is straightforward to check that 
$\Q\left(\tau_{\mathcal R}<t_0-t\right)=o(t_0-t)$. Since $x_2<0$, \eqref{E:cont} thus implies that
$v(t,x)<1$ for $t$ close to $t_0$, which contradicts \eqref{v1}.

The above argument also works to show that $h(T-)=0$.

\item
The proof of \eqref{E:PDEc} is along standard lines (e.g.~see \cite[Theorem 7.7 in Chapter 2]{KS2}), so we do not include it here. 

Let us next establish the smooth-fit property. Since $x\mapsto v(t,x)$ is non-decreasing, 
it suffices to show that 
\[
\lim_{\epsilon \downarrow 0} \frac{v(t, h(t)+\epsilon) -v(t, h(t))}{\epsilon} \leq 0.
\]
Without loss of generality, let $t=0$. Writing $x=h(0)$, it is enough to show that
\[
v(t, x+\epsilon) - v(t, x) = o(\epsilon) \quad \text{as } \epsilon \searrow 0.
\]
Denoting the optimal stopping time when starting at the point $(0, x+\epsilon)$ by $\tau_{\epsilon}$, we have
\begin{IEEEeqnarray*}{rCl}
v(t, x+\epsilon) - v(t, x) &\leq& \E^{\Q} \big[e^{\int_{0}^{\tau_{\epsilon}} \hat X^{0,x+\epsilon}_{u}\ud u} \big] -\E^{\Q}\big[e^{\int_{0}^{\tau_{\epsilon}} \hat X^{0,x}_{u} \ud u}\big] \\
&=& \E^{\Q}\big[e^{\int_{0}^{\tau_{\epsilon}} \hat X^{0,x+\epsilon}_{u} \ud u} (1 - e^{\int_{0}^{\tau_{\epsilon}} \hat X^{0,x}_{u} - \hat X^{0,x+\epsilon}_{u} \ud u}) \big] \\
&\leq& \E^{\Q}\big[e^{\int_{0}^{\tau_{\epsilon}} \hat X^{0,x+\epsilon}_{u} \ud u} \int_{0}^{\tau_{\epsilon}} \hat X^{0,x+\epsilon}_{u} - \hat X^{0,x}_{u} \ud u \big] \\
&\leq& \E^{\Q}\big[e^{\int_{0}^{\tau_{\epsilon}} \hat X^{0,x+\epsilon}_{u} \ud u} \big(\tau_{\epsilon}\int_{0}^{\tau_{\epsilon}} (\hat X^{0,x+\epsilon}_{u} - \hat X^{0,x}_{u})^{2} \ud u \big)^{1/2}\big] \\
&\leq& \E^{\Q}\big[\tau_{\epsilon} e^{2\int_{0}^{\tau_{\epsilon}} \hat X^{0,x+\epsilon}_{u} \ud u}  \big]^{1/2} \E^{\Q} \big[ \int_{0}^{\tau_{\epsilon}} ( \hat  X^{0,x+\epsilon}_{u} - \hat X^{0,x}_{u})^{2} \ud u \big]^{1/2},
\end{IEEEeqnarray*}
where the penultimate inequality follows from Jensen's inequality and the last one from Cauchy-Schwartz.
Since the boundary $h$ is non-decreasing, with the help of L\'evy's modulus of continuity theorem as well as the law of the iterated logarithm, we see that $\tau_{\epsilon} \tend 0$ a.s.~as $\epsilon \searrow 0$. Hence, by the dominated convergence theorem,
\[ \E^{\Q}[\tau_{\epsilon} e^{2\int_{0}^{\tau_{\epsilon}} \hat X^{0,x+\epsilon}_{u} \ud u}] \tend 0 \quad \text{as } \epsilon \searrow 0\]
with the dominating function as in the proof of Proposition \ref{T:OST}.

To complete the proof of smooth-fit, it suffices to show that  
\[
\E^{\Q}\lt[\lt(\int_{0}^{\tau_{\epsilon}} \hat X^{0,x+\epsilon}_{u} - \hat X^{0,x}_{u} \ud u \rt)^{2}\rt] = O(\epsilon^{2}) \quad \text{as } \epsilon \tend 0.
\]
To this end, 
\begin{IEEEeqnarray*}{rCl}
\E^{\Q}\lt[\lt(\int_{0}^{\tau_{\epsilon}} \hat X^{0,x+\epsilon}_{u} - \hat X^{0,x}_{u} \ud u \rt)^{2}\rt]
&\leq& T \E^{\Q}\lt[\int_{0}^{T} (\hat X^{0,x+\epsilon}_{u} - \hat X^{0,x}_{u})^{2} \ud u \rt] \\
&\leq& T^{2} \E^{\Q}\lt[\sup_{0\leq u \leq T} (\hat X^{0,x+\epsilon}_{u} - \hat X^{0,x}_{u})^{2}\rt]\\
&\leq& c\epsilon^{2},
\end{IEEEeqnarray*}
where $c$ is a constant dependent on $T$, $\sigma$, and the Lipschitz constant of $\psi$. In the above, the first inequality comes from Jensen's inequality, and the last inequality is a standard estimate coming from an application of Gronwall's inequality combined with Doob's $L^{2}$ inequality. This finishes the proof of the claim. 
\end{enumerate}
\end{proof}

\begin{SCfigure}[][h!]
\caption{The value function $v(t,x)$ in the case of a normal prior with standard deviation $\gm=0.5$, the market volatility $\sigma=0.2$.}
\includegraphics[width=0.55\textwidth]{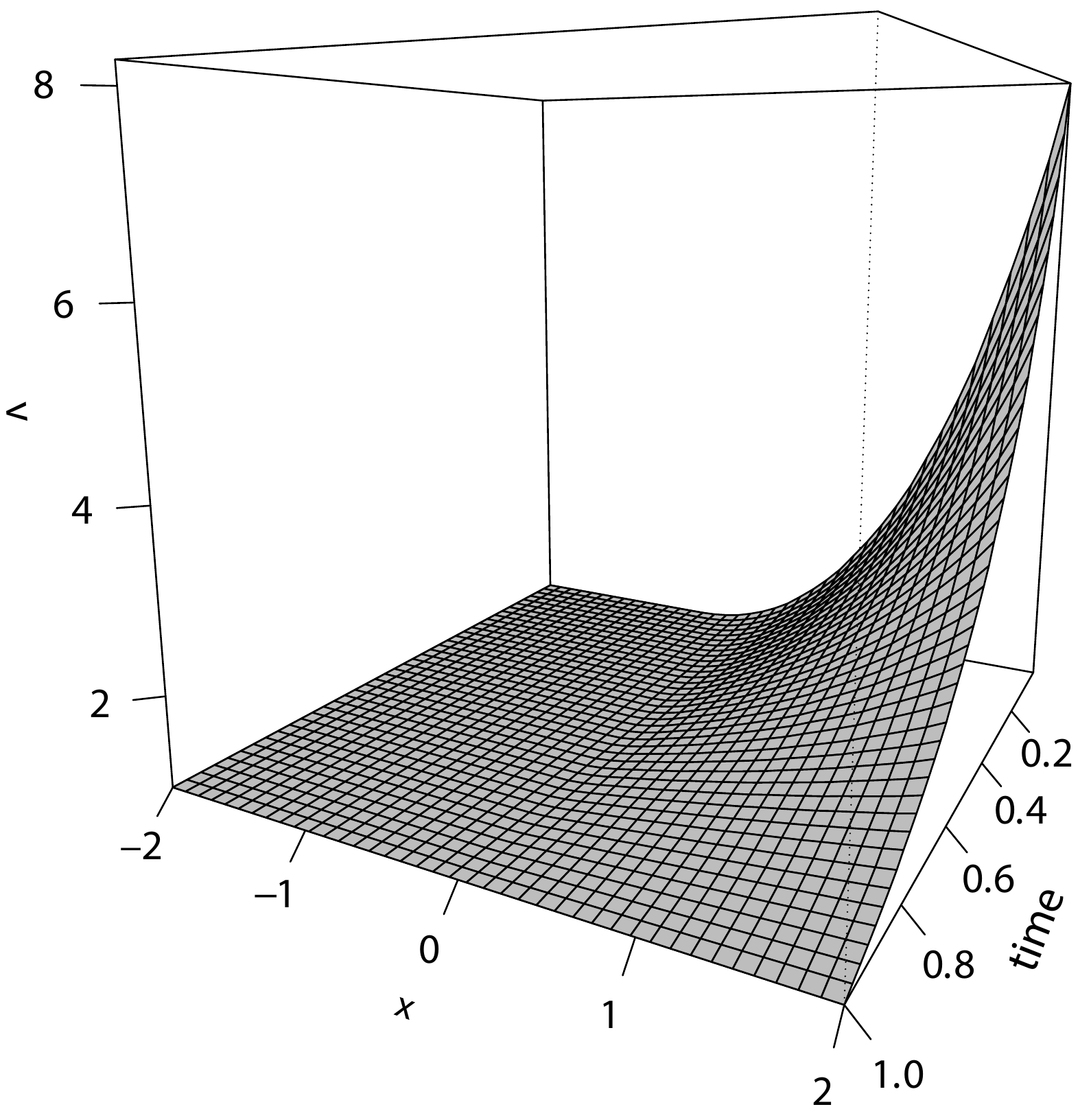}
\end{SCfigure}

\section{An integral equation for the boundary}

In this section we show that the optimal stopping boundary can be characterised as 
the unique solution of a non-linear integral equation. The proof follows along similar lines
as in \cite{J} and \cite{P2}.

\begin{theorem}[Optimal stopping boundary] 
The stopping boundary $h$ is the unique solution to the integral equation
\begin{IEEEeqnarray}{rCl} \label{E:IE}
\E^\Q \lt[ e^{\int_t^T  \hat X^{t, h(t)}_{u} \ud u } \rt] &=& 1 + \int_t^T \E^\Q\lt[  e^{\int_t^s  \hat X^{t, h(t)}_u \ud u } \hat X^{t, h(t)}_s 1_{\{ \hat X^{t,h(t)}_s \leq h(s)\}} \rt] \ud s 
\end{IEEEeqnarray}
in the class of non-positive continuous functions.
\end{theorem}

\begin{proof} 
An application of It\^{o}'s formula  (more precisely, its extension proved in \cite{P}, which can be applied thanks to the monotonicity of $h$)  
to $v(s, \hat X^{t,x}_s) e^{\int_t^s \hat X^{t,x}_u \ud u}$ yields
\begin{IEEEeqnarray}{rCl} \label{E:VITO}
v(s, \hat X^{t,x}_s) e^{\int_t^s  \hat X^{t,x}_u \ud u} &=&
v(t, \hat X^{t,x}_t) + \int_t^s e^{\int_t^r \hat X^{t,x}_u \ud u } \lt( \mathcal{L}_{\hat X^{t,x}} v(r, \hat X^{t,x}_r)  +  \hat  X^{t,x}_r v(r, \hat X^{t,x}_{r}) \rt) \ud r \nonumber \\ 
&&+ \int_t^s e^{\int_t^r  \hat X^{t,x}_u \ud u } \psi(r, \hat X^{t,x}_r) \pt_{2} v(r, \hat X^{t,x}_{r}) \ud Z_r \,. 
\end{IEEEeqnarray} 
Let us introduce a localising sequence  $\tau_{n}:= \inf \{ r \geq t \,:\,   \hat X^{t,x}_{r} \geq n \} \wedge T$; it satisfies $\tau_{n} \nearrow T$ a.s.~as $n \tend \infty$. Since, for all $n \in \N$,
\[
\E^{\Q} \lt[ \int_t^{s\wedge\tau_n} e^{\int_t^r  \hat X^{t,x}_u \ud u } \psi(r, \hat X^{t,x}_r) \pt_{2} v(r, \hat X^{t,x}_{r}) \ud Z_r \rt]=0,
\]
from \eqref{E:VITO} we get
\begin{IEEEeqnarray*}{rCl}
\E^\Q[v(T\wedge \tau_{n}, \hat X^{t,x}_{T\wedge \tau_{n}}) e^{\int_t^{T\wedge \tau_{n}} \hat X^{t,x}_u \ud u}] &=& v(t, x) +  \E^\Q \lt[ \int_t^{T\wedge \tau_{n}} e^{\int_t^r  \hat X^{t,x}_u \ud u }  \hat X^{t,x}_r \Ind_{\{\hat X^{t,x}_r \leq h(r)\}} \ud r \rt].
\end{IEEEeqnarray*}
Letting $n \tend \infty$, the equation becomes
\begin{IEEEeqnarray}{rCl}
 \label{E:VXT}
\E^\Q[v(T, \hat X^{t,x}_T) e^{\int_t^T \hat X^{t,x}_u \ud u}] &=& v(t, x) +  \int_t^T \E^\Q \lt[e^{\int_t^r  \hat X^{t,x}_u \ud u }  \hat X^{t,x}_r \Ind_{\{\hat X^{t,x}_r \leq h(r)\}} \rt]\ud r \,.
\end{IEEEeqnarray}
Here, the left-hand side is obtained by dominated convergence as \mbox{$v(T\wedge \tau_{n}, \hat X^{t,x}_{T\wedge \tau_{n}}) e^{\int_t^{T\wedge \tau_{n}} \hat X^{t,x}_u \ud u}$} is dominated by $e^{2T( \sup_{t\leq u \leq T} \hat X^{t,x}_u\vee 0)}$, which is integrable; the right-hand side comes from monotone convergence. Substitution $x=h(t)$ in \eqref{E:VXT} gives 
\begin{IEEEeqnarray*}{rCl}
\E^\Q[e^{\int_t^T  \hat X^{t,h(t)}_u \ud u}] &=& 1 +  \int_t^T \E^\Q \lt[e^{\int_t^r  \hat X^{t,h(t)}_u \ud u }  \hat X^{t,h(t)}_r \Ind_{\{\hat X^{t,h(t)}_r \leq h(r)\}} \rt] \ud r \,,
\end{IEEEeqnarray*}
which shows that $h$ solves the integral equation \eqref{E:IE}.

For uniqueness, assume that $t \mapsto k(t)$ is another non-positive continuous solution to \eqref{E:IE} and define
\begin{IEEEeqnarray}{rCl}
\label{E:vtilde}
\tilde{v}(t, x) := \E^\Q [e^{\int_t^T  \hat X^{t,x}_u \ud u} ] - \E^\Q \lt[ \int_t^T e^{\int_t^r  \hat X^{t,x}_u \ud u}  \hat X^{t,x}_r \Ind_{\{ \hat X^{t,x}_r \leq k(r) \}} \ud r \rt].
\end{IEEEeqnarray}
Using \eqref{E:VXT}, \eqref{E:vtilde} and the Markov property, the two processes defined for $s\in [t, T]$ as 
\begin{IEEEeqnarray*}{rCl}
M^{\tilde{v}}_s := \tilde{v}(s, \hat X^{t,x}_ s ) e^{\int_t^s  \hat X^{t,x}_u \ud u} - \int_t^s e^{\int_t^r  \hat X^{t,x}_u \ud u}  \hat X^{t,x}_r \Ind_{\{ \hat X^{t,x}_r \leq k(r) \}} \ud r
\end{IEEEeqnarray*}
and
\begin{IEEEeqnarray*}{rCl}
M^v_s := v(s, \hat X^{t,x}_ s ) e^{\int_t^s  \hat X^{t,x}_u \ud u} - \int_t^s e^{\int_t^r  \hat X^{t,x}_u \ud u}  \hat X^{t,x}_r \Ind_{\{ \hat X^{t,x}_r \leq h(r) \}} \ud r
\end{IEEEeqnarray*}
are easily verified to be $\Q$-martingales.

\noindent 
\underline{Claim 1}: $\tilde{v}(t,x) = 1$ for $x\leq k(t)$. 

Let $x \leq k(t)$ and define $\gamma_k := \inf \{ s \geq 0 : \hat X^{t,x}_{t+s} \geq k(t+s) \} \wedge (T-t)$. Then 
\begin{IEEEeqnarray*}{rCl}
M^{\tilde v}_{t + \gamma_k} &=&  \tilde{v}(t+\gamma_k, \hat X^{t,x}_{t+\gamma_k} ) e^{\int_t^{t+\gamma_k}  \hat X^{t,x}_u \ud u} - \int_t^{t+\gamma_k} e^{\int_t^r  \hat X^{t,x}_u \ud u}  \hat X^{t,x}_r \Ind_{\{ \hat X^{t,x}_r \leq k(r) \}} \ud r\\
&=& e^{\int_t^{t+\gamma_k}  \hat X^{t,x}_u \ud u} - \int_t^{t+\gamma_k} e^{\int_t^r  \hat X^{t,x}_u \ud u}  \hat X^{t,x}_r \ud r \\
&=& 1,
\end{IEEEeqnarray*}
where the second equality follows from \eqref{E:IE}. By optional sampling, 
\begin{IEEEeqnarray*}{rCl}
\tilde{v}(t,x) &=&  M^{\tilde v}_t = \E^\Q[M^{\tilde v}_{t + \gamma_k}]=1,
\end{IEEEeqnarray*}
which finishes the proof of Claim 1.

\noindent
\underline{Claim 2}: $\tilde{v} \leq v$. 

Suppose $x > k(t)$ and define $\tau_k:=\inf \{s \geq 0 \,: \hat X^{t,x}_{t+s} \leq k(t+s) \} \wedge (T-t)$. Then 
\begin{IEEEeqnarray*}{rCl}
\tilde{v}(t,x) &=& \E^\Q \lt[ \tilde{v}(t+\tau_k, \hat X^{t,x}_{t+\tau_k}) e^{\int_t^{t+\tau_k}  \hat X^{t,x}_u \ud u} \rt] \\
&=& \E^\Q \lt[ e^{\int_t^{t+\tau_k}  \hat X^{t,x}_u} \ud u \rt] \\
&\leq & v(t,x),
\end{IEEEeqnarray*}
with the first equality following from the martingality of $M^{\tilde{v}}$ and the optional sampling theorem, the second by the definition of $\tilde{v}$ and \eqref{E:IE}. Combining this with Claim 1, the result is obtained. 

\noindent
\underline{Claim 3}: $h \leq k$.

Assume for a contradiction that $h(t) > k(t)$ for some $t$. 
Let $x=k(t)$ and define $\gamma_h := \inf \{ s \geq 0 : \hat X^{t,x}_{t+s} \geq h(t+s) \} \wedge (T-t)$.  Then 
\begin{IEEEeqnarray*}{rCl}
0 &=& v(t, x)- \tilde{v}(t,x) \\
&=& \E^\Q[ M^v_{t+\gamma_h} - M^{\tilde{v}}_{t+\gamma_h}] \\
&=& \E^\Q \lt[ e^{\int_t^{t+\gamma_h}  \hat X^{t,x}_u \ud u} (v(t+\gamma_h, \hat X^{t,x}_{t+\gamma_h}) - \tilde{v}(t+\gamma_h, \hat X^{t,x}_{t+\gamma_h}) ) \rt]  \\
&&- \E^\Q \lt[ \int_t^{t+\gamma_h} e^{\int_t^{r}  \hat X^{t,x}_u \ud u}  \hat X^{t,x}_r \Ind_{\{ X^{t,x}_r \in (k(r), h(r)]\}} \ud r\rt].
\end{IEEEeqnarray*}
In the first equality above, $v(t,x)=1$ by the assumption  $h(t)> k(t)$, and $\tilde{v}(t,x)=1$ by the definition of $\tilde{v}$ and \eqref{E:IE}. The second equality comes from optional sampling. 
In the final expression, the first term is non-negative by Claim 2, the second term (including the minus sign in front) is strictly positive by the assumption $h(t)> k(t)$ together with the continuity and the non-positivity of $k$ and $h$. Hence we have obtained a contradiction,
which proves the claim.

It follows from \eqref{E:VXT}, \eqref{E:vtilde}, Claim 2, and Claim 3 that $v=\tilde v$.
Since $v(t,x)>1$ for $x>h(t)$, Claim 1 yields that $h\geq k$. In view of Claim 3, this finishes the proof.
\end{proof}

%
%


\section{Parameter dependence}

\subsection{Dependence of the value function on the market volatility}

A large volatility $\sigma$ makes the observation process noisy, slowing down the speed of 
learning about the drift. Since the fluctuations are trend-free, the intuition is that the agent should benefit from a smaller market volatility $\sigma$. While a full proof of this intuitive remark appears to be challenging, we have the following sufficient condition which guarantees monotonicity in $\sigma$.

\begin{theorem}
\label{monotonicity}
Assume that the dispersion function $\psi$ is such that $\sigma\psi(t,x)$ is non-increasing in $\sigma$. 
Then the value $V$ in \eqref{E:OS} is non-increasing in $\sigma$.
\end{theorem}

\begin{proof}
If $\sigma\psi(t,x)$ is non-increasing in $\sigma$, then both the drift term and the diffusion term of $\hat X$ are non-increasing in $\sigma$.
Therefore, Theorem 6.1 from \cite{ET08} can be applied to prove that the value function $v$ is decreasing in $\sigma$.
\end{proof}

\begin{example}[Two-point distribution]
Suppose $X$ has a two-point prior distribution $\mu = (1 - \pi) \delta_l + \pi \delta_h$, where $l< h$. 
Then $\sigma \psi(t,x)=(h-x) (x - l)$, so $V$ is decreasing in $\sigma$. 
\end{example}

\begin{example}[Normal distribution]
Suppose the prior distribution of $X$ is $N(m, \gamma^2)$. Then $\sigma\psi(t,x)=\frac{\sigma^2 \gamma^2}{\sigma^2+ \gamma^2 t}$,
which is {\em increasing} in $\sigma$. Thus Theorem~\ref{monotonicity} does not apply.
\end{example}

The difficulty in proving the intuitive conjecture that the initial value $V$ in \eqref{E:OS} should be decreasing in the volatility $\sigma$ lies in the fact that it is not true in general that the Markovian value function $v$ in \eqref{E:VM} is decreasing in $\sigma$. We can see this in the case of a normal prior in Figure \ref{F:Diff}. The picture depicts the difference between two Markovian value functions for the same normal prior with standard deviation $\gm=0.5$, but different volatilities $\sigma$. Nevertheless, the same picture shows that at time $t=0$, the difference is positive, so conforming with our intuitive conjecture.

\begin{SCfigure}[][h!] \label{F:Diff}
\centering
\caption{The difference $v_{0.2}-v_{0.5}$ between two value functions; $v_{0.2}$ and $v_{0.5}$ denote the value functions in the cases of the market volatility $\sigma$ being equal to $0.2$ and $0.5$, respectively. }
\includegraphics[width=0.6\textwidth]{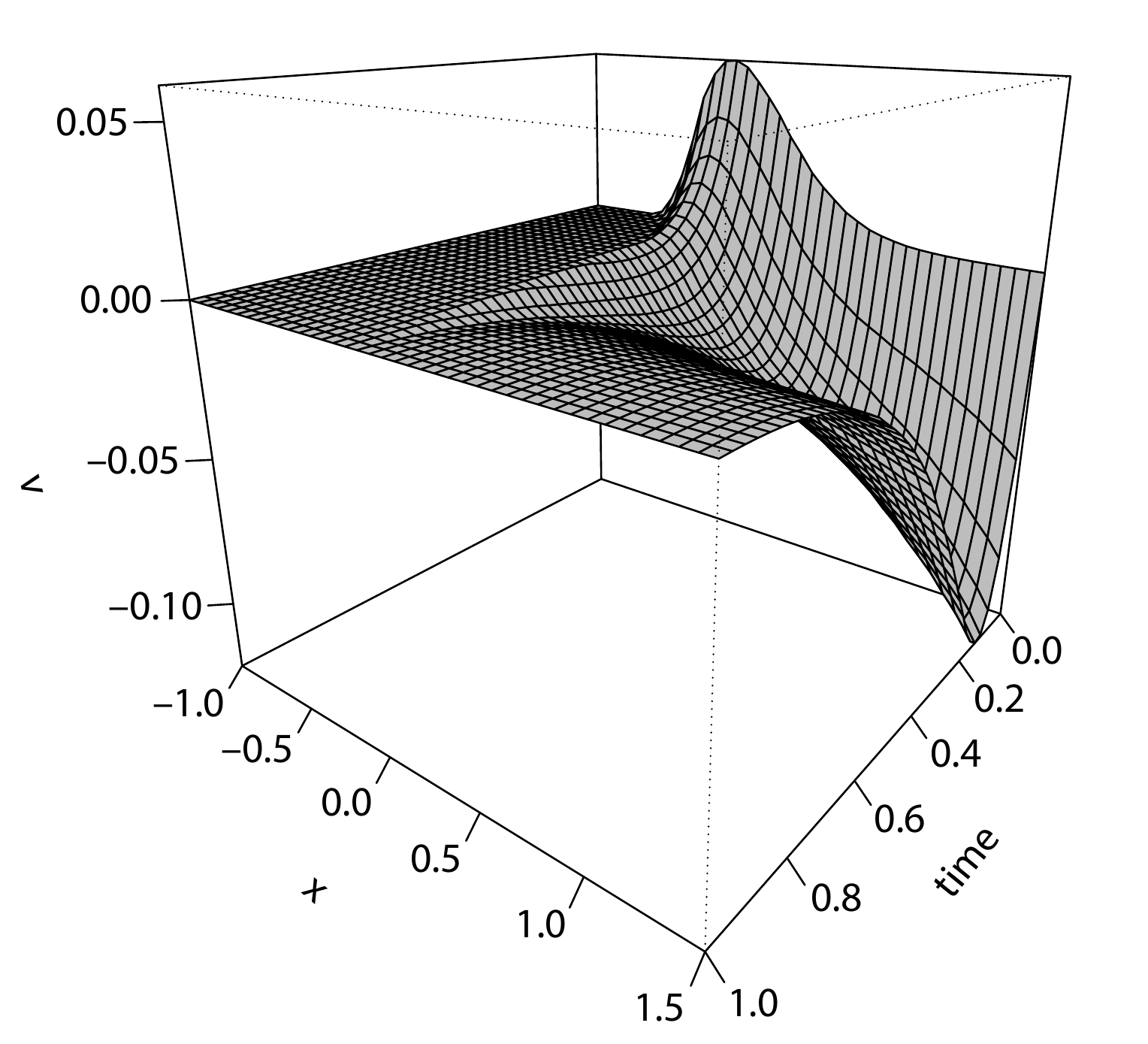}
\end{SCfigure}

As far as the optimal stopping boundaries are concerned, the lack of monotonicity of the Markovian value function in the volatility $\sigma$ manifests in that the stopping boundaries for different values of $\sigma$ may intersect. An example of this appears in Figure~\ref{F:SigB}. The same graph also provides intuition about how the shape of the boundary changes as one varies the parameter $\sigma$. In particular, we get an impression what boundary to expect as $\sigma$ approaches zero or grows to infinity.

\begin{SCfigure}[][h] \label{F:SigB}
\includegraphics[width=0.65\textwidth,clip, trim=6mm 9mm 15mm 0mm]{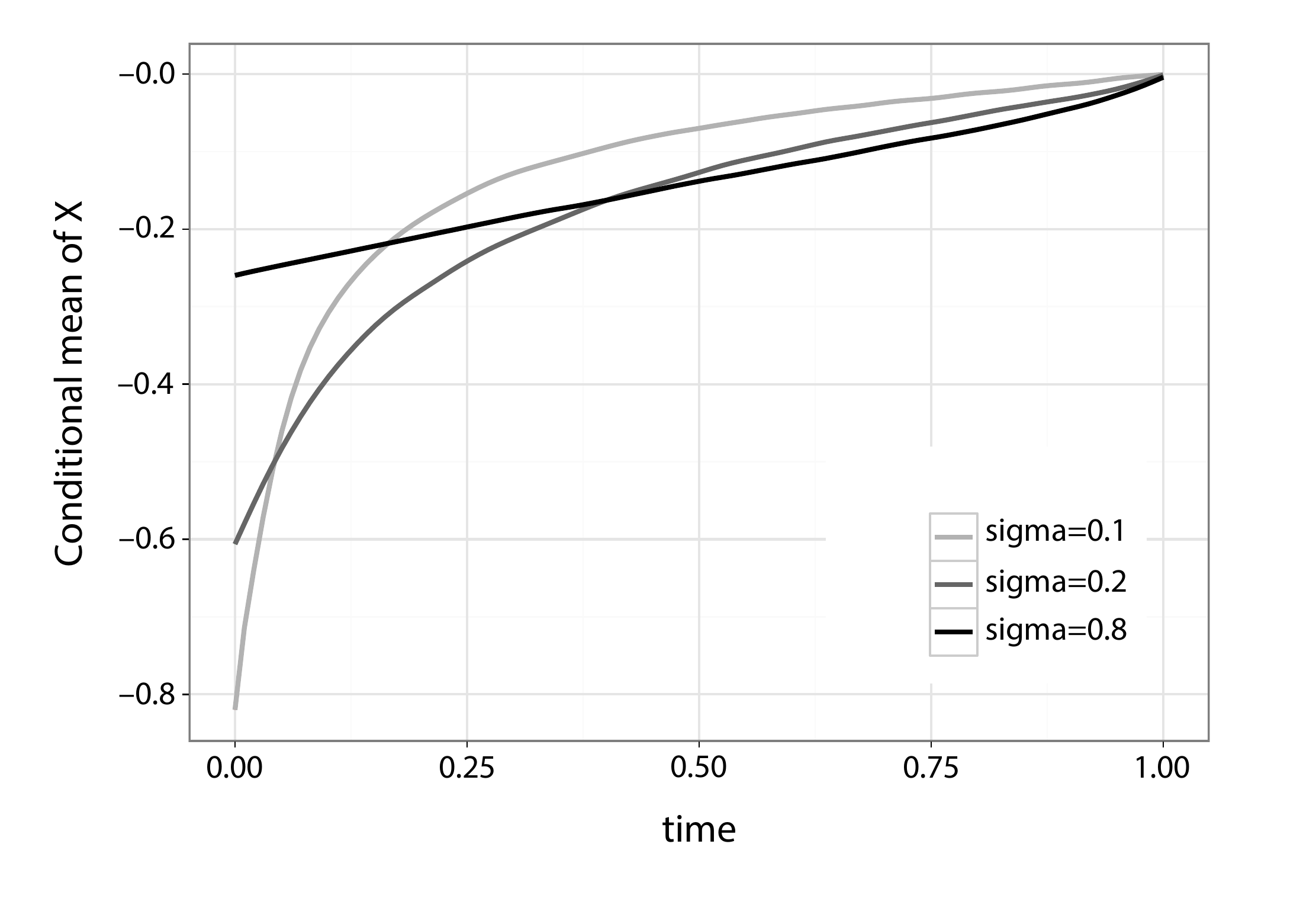}
\caption{Optimal stopping boundaries for different values of market volatility $\sigma$ in the case a normal prior with standard deviation $\gamma=0.5$.}
\end{SCfigure}

\subsection{Dependence of the value function on the initial prior}

\begin{theorem} \label{T:PrM}
Assume that $\mu_1$ and $\mu_2$ are two prior distributions such that the corresponding volatilities $\psi_1$ and $\psi_2$ 
satisfy $\psi_1(t,x)\leq \psi_2(t,x)$ for all $(t,x)\in[0,T]\times \R$. Then the corresponding Markovian value functions
$v_1$ and $v_2$ satisfy $v_1\leq v_2$.
\end{theorem}

\begin{proof}
Again, Theorem 6.1 from \cite{ET08} can be applied to prove that the value function $v$ is increasing in $\psi$.
\end{proof}

In the case of the normal prior, the function $\psi(t,x) = \frac{\sigma\gamma^2}{\sigma^2+t\gamma^2}$ is monotonically increasing in the standard deviation $\gamma$ of the prior. Hence Theorem \ref{T:PrM} applies and the Markovian value function $v$ increases in $\gamma$. A consequence of this is that optimal stopping boundaries are ordered by the size of $\gamma$ as shown in Figure \ref{F:GB}.

\begin{SCfigure}[][h] \label{F:GB}
\includegraphics[width=0.65\textwidth,clip, trim=6mm 8mm 15mm 0mm]{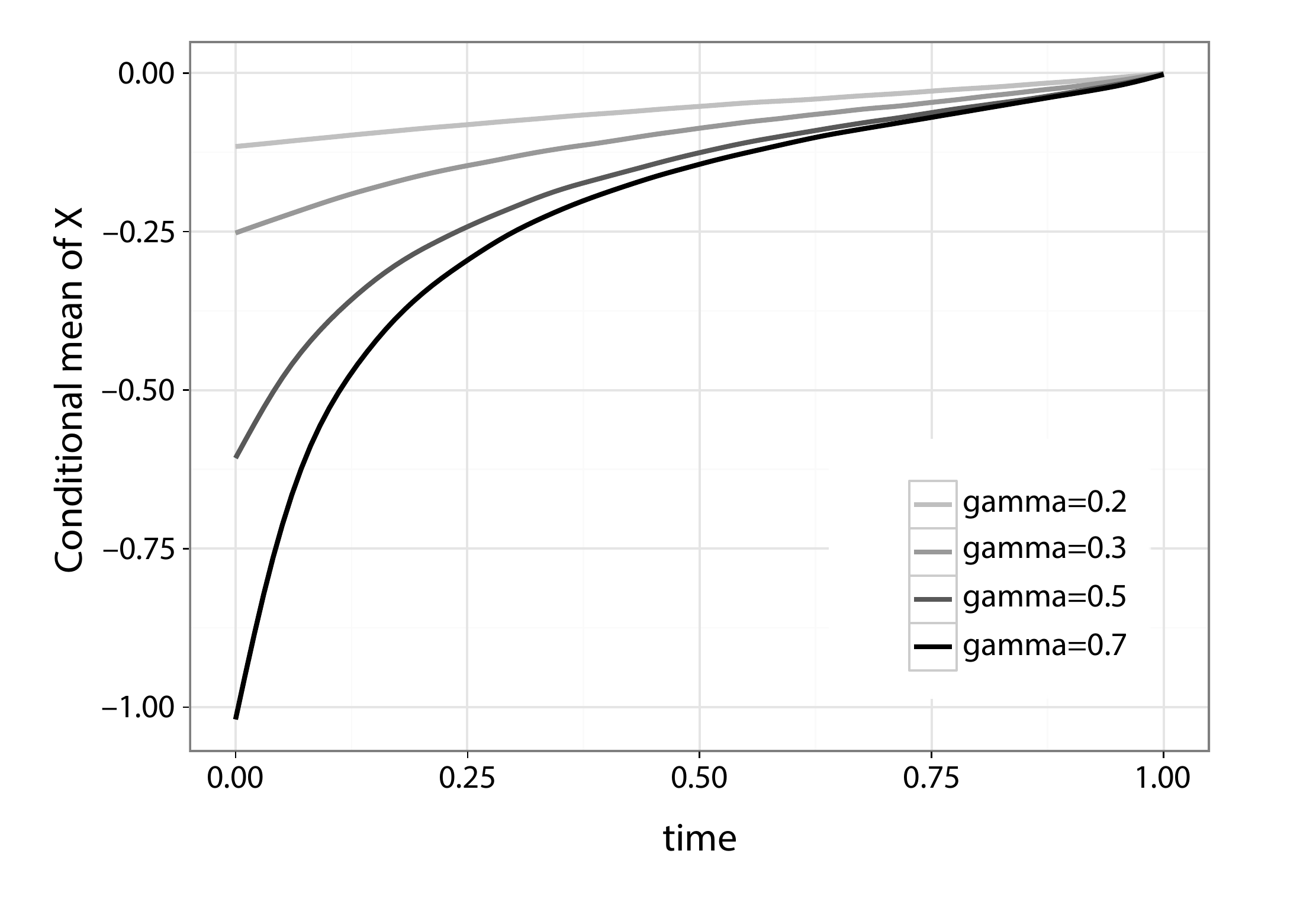}
\caption{Optimal stopping boundaries for different values of standard deviation $\gamma$ in the case of a normal prior when the market volatility $\sigma=0.2$.}
\end{SCfigure}

For compactly supported distributions, Theorem \ref{T:PrM} offers a way to construct an upper bound for the Markovian value function $v$. Suppose the prior $\mu$ is a compactly supported distribution. Since $\psi_{\mu}$ is bounded, by the two-point prior example on page \pageref{Ex:B}, we can find a two-point distribution $\eta := (1-\pi)\delta_{a} + \pi \delta_{b}$ with $\int_{\R} u \mu(\ud u) = \int_{\R} u \eta(\ud u)$ such that $\psi_{\mu} \leq \psi_{\eta}$. Then Theorem \ref{T:PrM} yields that $v_{\mu} \leq v_{\eta}$ and so the stopping boundaries satisfy $h_{\eta} \leq h_{\mu}$. As a result, $V_{\mu} \leq \E_{\eta}[S_{\tau_{h_{\eta}}}]$, where $V_{\mu}$ denotes the initial value for the prior $\mu$, $\E_{\eta}$ denotes the expectation operator under which the prior is $\eta$ instead of $\mu$.


\subsection{The value of filtering: numerical investigation}

Having introduced and solved the optimal liquidation problem for an arbitrary prior, a pragmatic question arises - how much is there to be gained from the elaborate sequential liquidation strategy with real-time filtering in comparison to a naive optimal selling strategy without filtering? In this part, we address the question within our model from a numerical point of view in the normal prior case.

Let us consider an agent who wants to liquidate an asset before time $T$. We suppose the asset price evolves according to \eqref{E:S} and that the agent's prior $\mu$ for the drift $X$ is the normal distribution $N(m, \gm)$. If the agent does not know about the possibility of real-time filtering, he does not utilise any valuable information from the asset price observations and so, at time $0$, will make a decision whether to sell immediately, i.e.~at time $0$, or liquidate at time $T$. The expected value from an optimal liquidation strategy with selling allowed only at times $0$ and $T$ is 
\[
V_{\{0,T\}}:=\E[e^{XT}]\vee 1= e^{mT + \frac{1}{2}(\gm T)^{2}}\vee 1\,.
\]
However, if the agent is aware of an optimal sequential liquidation procedure involving filtering, the expected value from an optimal selling is
 \[
V=\sup_{\tau \in \mathcal{T}_T^{S} } \E[S_\tau].
\]
It is worth noting that the inequality $V_{\{0,T\}}\leq V$ holds for any prior $\mu$, thus giving us a lower bound for the value $V$.

In Figure \ref{F:VC}, we see the two values $V_{\{0,T\}}$ and $V$ calculated for two different priors at a range of different market volatilities. In addition, Figure \ref{F:Imp} depicts the percentage improvement $(V-V_{\{0,T\}})/V_{\{0,T\}}$ that the sequential procedure with filtering brings over the naive strategy. Dependence of the two values on the standard deviation $\gamma$ of the normal prior is illustrated in Figure \ref{F:DG}.


\begin{SCfigure}[][h] \label{F:VC}
\includegraphics[width=0.6\textwidth, clip, trim=10mm 0mm 14mm 0mm]{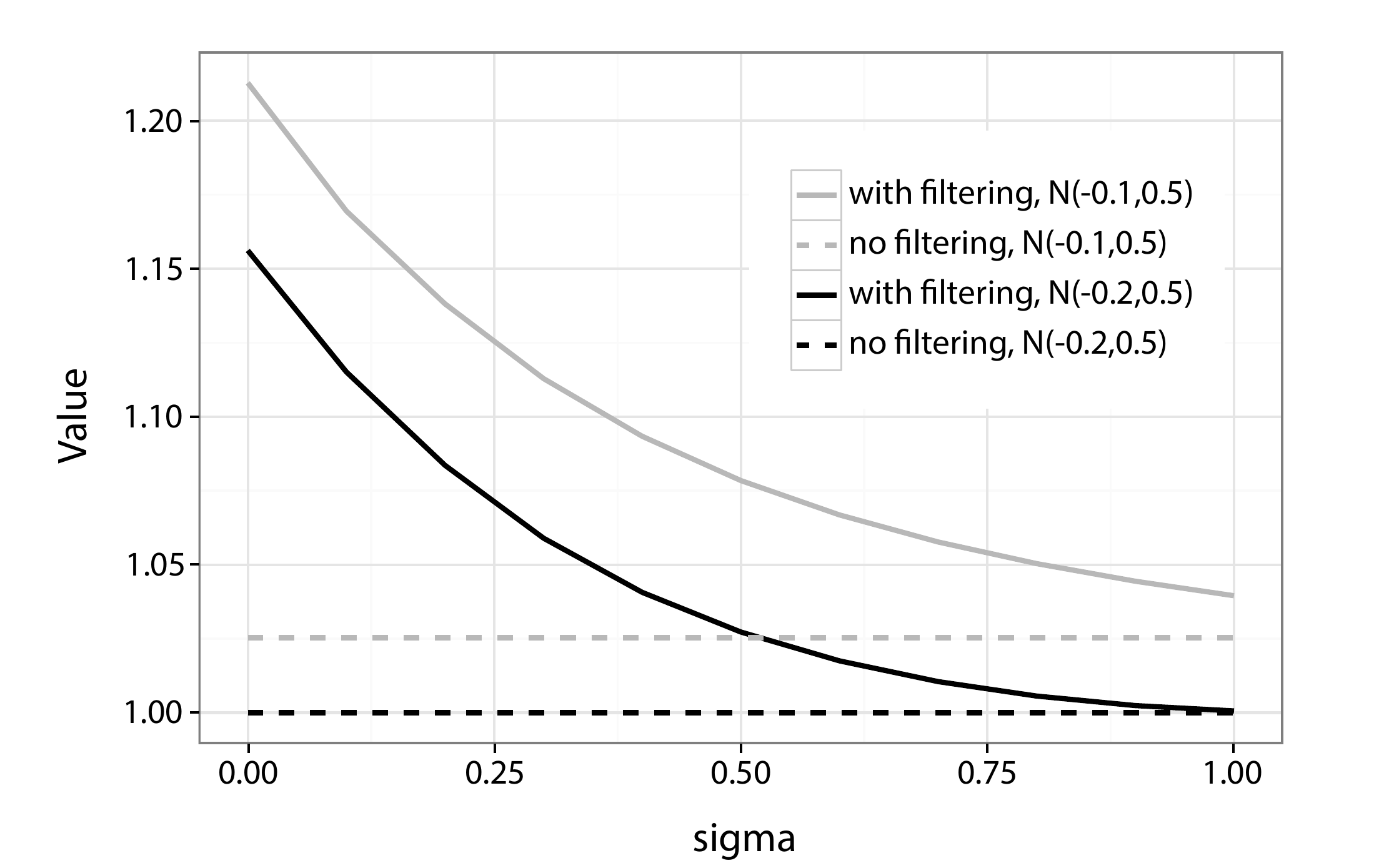}
\caption{{\bf Initial value as a function of market volatility.} The solid gray curve corresponds to the optimal liquidation value and the dashed gray line to the value without filtering - both for the normal prior $N(-0.1, 0.5)$. Similarly, the solid black curve corresponds to the optimal liquidation value and the dashed black line to the value without filtering for the normal prior $N(-0.2, 0.5)$. }
\end{SCfigure}

\begin{SCfigure}[][h!] \label{F:Imp}
\includegraphics[width=0.6\textwidth, clip, trim=8mm 11mm 17mm 0mm]{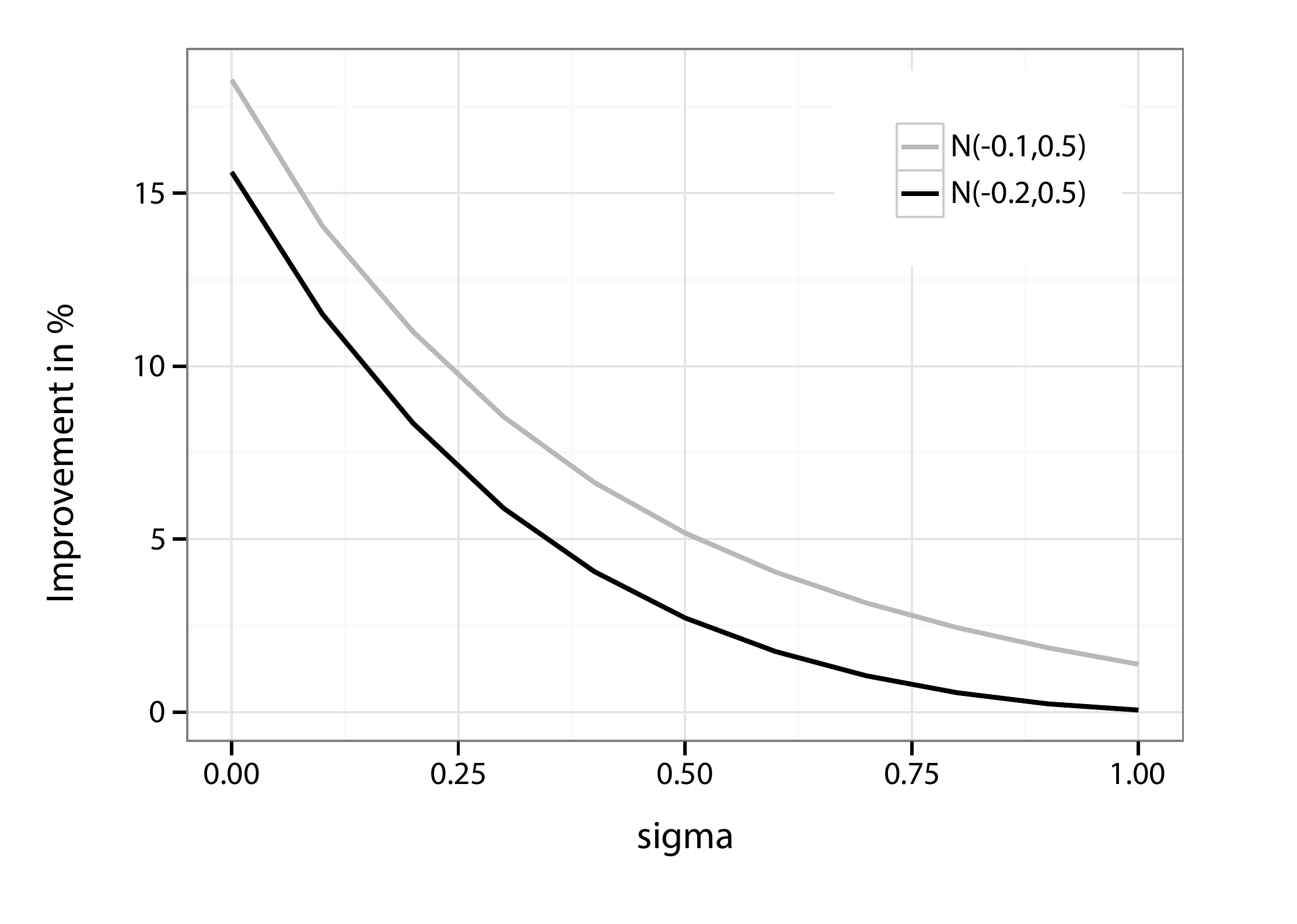}
\caption{ {\bf Improvement due to filtering}. The percentage improvement $(V-V_{\{0,T\}})/V_{\{0,T\}}$ over the strategy without filtering.}
\end{SCfigure}

\begin{SCfigure}[][h!] \label{F:DG}
\includegraphics[width=0.6\textwidth, trim=10mm 5mm 16mm 2mm, clip]{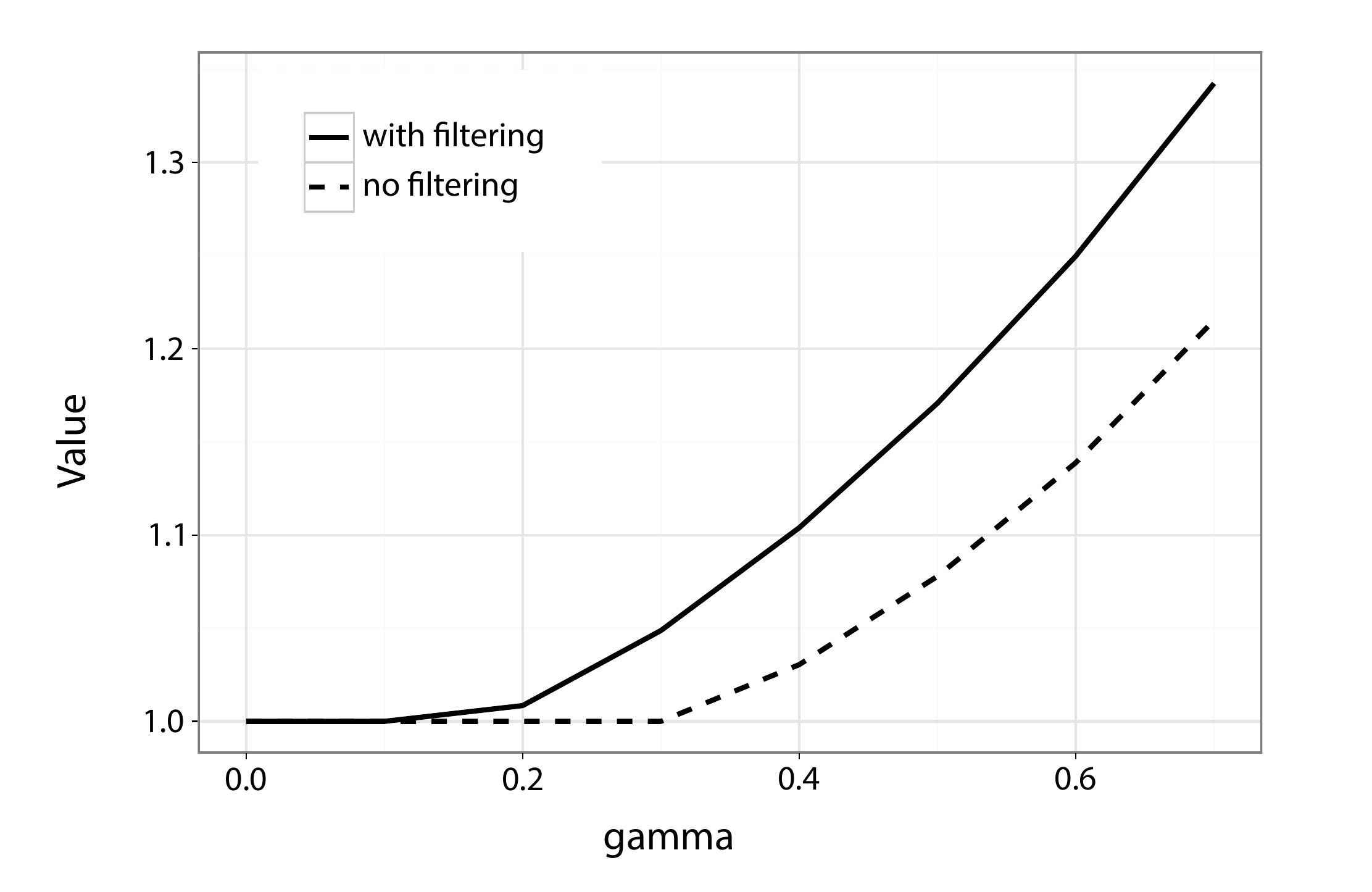}
\caption{{\bf Initial value as a function of standard deviation of the normal prior.} The solid curve corresponds to the optimal liquidation value while the dashed line to the value without filtering. Here the prior is normal with mean $-0.05$, the market volatility $\sigma=0.2$.}
\end{SCfigure}

\end{document}